\documentclass{sig-alternate-10pt}

\usepackage{cite}
\usepackage[usenames]{color}
\usepackage[noend]{algorithmic}
\usepackage{algorithm}
\usepackage{stmaryrd}
\usepackage{subfigure}
\usepackage{multicol}
\usepackage{paralist}
\usepackage{pictexwd}
\usepackage{xspace}
\usepackage{balance}
\usepackage{times}

\def\e#1{\emph{#1}}

\newcommand{\eat}[1]{}

\newcommand{\set}[1]{\{#1\}}

\def\eqdef{\stackrel{\textsf{\tiny def}}{=}}

\newcommand{\algname}[1]{{\sf #1}}
\def\myrulewidth{3.20in}

\def\therule{\rule{\myrulewidth}{0.2pt}}

\newenvironment{algseries}[2]
{\centering\begin{figure}[#1]\begin{center}\def\thecaption{\caption{#2}}
			\vskip-0.8em\begin{tabular}{p{\myrulewidth}}\therule\end{tabular}\vskip0.2em}
		{\thecaption \end{center}\end{figure}}

\newenvironment{insidealg}[2]
{\normalsize\begin{insidecode}{#1}{#2}{Algorithm}}
	{\end{insidecode}}

\newenvironment{insidecode}[3]
{
	\begin{tabular}{p{\myrulewidth}}
		\multicolumn{1}{c}{\rule{0mm}{3mm}{\bf #3} $\algname{#1}(\mbox{#2})$\vspace{-0.6em}}\\
		\therule\vskip-0.8em\therule
		\vspace{-1em}
		\begin{algorithmic}[1]}
		{\end{algorithmic}
		\vskip-0.3em\therule
	\end{tabular}}

	\def\arity{\mathord{\mathrm{arity}}}
	\def\consts{\mathsf{Const}}
	\newcommand{\sch}[1]{\mathbf{#1}}
	\newcommand{\scs}{\sch{S}}
	\newcommand{\graph}{\mathcal{H}}
	\newcommand{\relset}{\mathcal{R}}
	
	\newcommand{\depset}{\mathrm{\Delta}}
	\newcommand{\tup}[1]{\mathbf{#1}}
	
	\newcommand{\angs}[1]{\mathord{\langle #1 \rangle}}
	
	\def\phi{\varphi}
	\def\H{\mathcal{H}}

	\newtheorem{theorem}{Theorem}[section]
	\newtheorem{corollary}[theorem]{Corollary}
	\newtheorem{proposition}[theorem]{Proposition}
	\newtheorem{lemma}[theorem]{Lemma}
	
	\newtheorem{problem}[theorem]{Problem}

	\newtheorem{definitionthm}[theorem]{Definition}

	\newenvironment{repeatresult}[2]
	{\vskip0.5em\par\textsc{#1 #2.}\em}
	{\vskip1em}

	\newenvironment{reptheorem}[1]{\begin{repeatresult}{Theorem}{#1}}{\end{repeatresult}}
	\newenvironment{replemma}[1]{\begin{repeatresult}{Lemma}{#1}}{\end{repeatresult}}

	\newtheorem{examplethm}[theorem]{Example}
	\newenvironment{example}{\begin{examplethm}\em}
		{\qed\end{examplethm}}
	
	\newtheorem{commentthm}[theorem]{Comment}
	
	\newenvironment{commentqed}{\begin{commentthm}\em}{\qed\end{commentthm}}

	\newenvironment{citedtheorem}[1]
	{\begin{theorem}\hskip-0.2em\e{\cite{#1}}\,\,}
		{\end{theorem}}
	
	\newenvironment{citedproposition}[1]
	{\begin{proposition}\hskip-0.2em\e{\cite{#1}}\,\,}
		{\end{proposition}}

	\def\piptwo{\Pi_2^{\mathsf{p}}}

	\def\inds#1{\llbracket #1 \rrbracket}
	
	\def\partitle#1{\vskip1em\par\noindent\textbf{#1.}\,\,}
	\def\ititle#1{\textbf{#1}\ }

	\def\val#1{\mathsf{#1}}

	\def\rep#1{\mathrm{Rep}(#1)}
	\def\prep#1{\mathrm{PRep}(#1)}
	\def\grep#1{\mathrm{GRep}(#1)}
	\def\crep#1{\mathrm{CRep}(#1)}
	
	\def\succs{\succ^{\!\!+}}
	\def\succssubs{\succs}

	\def\rel#1{\mathit{#1}}
	\def\att#1{\mathrm{#1}}
	\def\qsattwo{\mathsf{QCNF}_2}
	
	\newenvironment{citemize}
	{\begin{compactitem}}
		{\end{compactitem}}
	
	\newenvironment{cenumerate}
	{\begin{compactenum}}
		{\end{compactenum}}

\begin{document}

\conferenceinfo{??}{??}

\pagestyle{plain}
\pagenumbering{arabic}

\title{Unambiguous Prioritized Repairing of Databases\titlenote{This
    research is supported by the Israeli Science Foundation, Grant
    \#1295/15.}}  \numberofauthors{3} \author{
  \alignauthor Benny Kimelfeld\titlenote{Taub Fellow, supported by the Taub Foundation.}\\
  \affaddr{Technion}\\
  \affaddr{Haifa 32000, Israel} \email{bennyk@cs.technion.ac.il}
\alignauthor Ester Livshits\\
\affaddr{Technion}\\
\affaddr{Haifa 32000, Israel}
\email{esterliv@cs.technion.ac.il}
\alignauthor Liat Peterfreund\\
\affaddr{Technion}\\
\affaddr{Haifa 32000, Israel} \email{liatpf@cs.technion.ac.il} }


\maketitle

\begin{abstract}
  In its traditional definition, a repair of an inconsistent database
  is a consistent database that differs from the inconsistent one in a
  ``minimal way.'' Often, repairs are not equally legitimate, as it is
  desired to prefer one over another; for example, one fact is
  regarded more reliable than another, or a more recent fact should be
  preferred to an earlier one. Motivated by these considerations,
  researchers have introduced and investigated the framework of
  preferred repairs, in the context of denial constraints and subset
  repairs. There, a priority relation between facts is lifted towards
  a priority relation between consistent databases, and repairs are
  restricted to the ones that are optimal in the lifted sense. Three
  notions of lifting (and optimal repairs) have been proposed: Pareto,
  global, and completion.

  In this paper we investigate the complexity of deciding whether the
  priority relation suffices to clean the database unambiguously, or
  in other words, whether there is exactly one optimal repair. We show
  that the different lifting semantics entail highly different
  complexities. Under Pareto optimality, the problem is coNP-complete,
  in data complexity, for every set of functional dependencies (FDs),
  except for the tractable case of (equivalence to) one FD per
  relation. Under global optimality, one FD per relation is still
  tractable, but we establish $\piptwo$-completeness for a relation
  with two FDs. In contrast, under completion optimality the problem
  is solvable in polynomial time for every set of FDs. In fact, we
  present a polynomial-time algorithm for arbitrary conflict
  hypergraphs. We further show that under a general assumption of
  transitivity, this algorithm solves the problem even for global
  optimality. The algorithm is extremely simple, but its proof of
  correctness is quite intricate.
\end{abstract}


\section{Introduction}
Managing database inconsistency has received a lot of attention in the
past two decades. Inconsistency arises for different reasons and in
different applications. For example, in common applications of Big
Data, information is obtained from imprecise sources (e.g., social
encyclopedias or social networks) via imprecise procedures (e.g.,
natural-language processing). It may also arise when integrating
conflicting data from different sources (each of which can be
consistent).  Arenas, Bertossi and
Chomicki~\cite{DBLP:conf/pods/ArenasBC99} introduced a principled
approach to managing of inconsistency, via the notions of \e{repairs}
and \e{consistent query answering}.  Informally, a \emph{repair} of an
inconsistent database $I$ is a consistent database $J$ that differs
from $I$ in a ``minimal'' way, where \e{minimality} refers to the
\e{symmetric difference}. In the case of anti-symmetric integrity
constraints (e.g., denial constraints and the special case of
functional dependencies), such a repair is a \e{subset repair} (i.e.,
$J$ is a consistent subinstance of $I$ that is not properly contained
in any consistent subinstance of $I$).

Various computational problems around database repairs have been
extensively investigated.  Most studied is the problem of computing
the \emph{consistent answers} of a query $q$ on an inconsistent
database $I$; these are the tuples in the intersection $\bigcap
\{q(J): \mbox{$J$ is a repair of
	$I$}\}$~\cite{DBLP:conf/pods/ArenasBC99,DBLP:conf/pods/KoutrisW15}. Hence,
in this approach inconsistency is handled at query time by returning
the tuples that are guaranteed to be in the result no matter which
repair is selected.  Another well studied question is that of
\e{repair checking}~\cite{DBLP:conf/icdt/AfratiK09}: given instances
$I$ and $J$, determine whether $J$ is a repair of $I$.  Depending on
the type of repairs and the type of integrity constraints, these
problems may vary from tractable to highly intractable complexity
classes.  See~\cite{Bertossi2011} for an overview of results.

In the above framework, all repairs of a given database instance are
taken into account, and they are treated on a par with each
other. There are situations, however, in which it is natural to prefer
one repair over
another~\cite{DBLP:journals/tods/FanGW12,DBLP:conf/sigmod/CaoFY13,DBLP:conf/edbtw/StaworkoCM06,DBLP:journals/amai/StaworkoCM12}. For
example, this is the case if one source is regarded to be more
reliable than another (e.g., enterprise data vs.~Internet harvesting,
precise vs.~imprecise sensing equipment, etc.)  or if available
timestamp information implies that a more recent fact should be
preferred over an earlier fact. Recency may be implied not only by
timestamps, but also by evolution semantics; for example, ``divorced''
is likely to be more updated than ``single,'' and similarly is
``Sergeant'' compared to
``Private.''\footnote{See~\cite{DBLP:series/synthesis/2012Fan} for a
	survey on aspects of data quality.}  Motivated by these
considerations, Staworko, Chomicki and
Marcinkowski~\cite{DBLP:conf/edbtw/StaworkoCM06,DBLP:journals/amai/StaworkoCM12}
introduced the framework of \e{preferred} repairs. The main
characteristic of this framework is that it uses a \e{priority}
relation between conflicting facts of an inconsistent database to
define a notion of \e{preferred} repairs.

Specifically, the notion of \e{Pareto optimality} and that of  \e{global
	optimality} are based on two different notions of
\e{improvement}---the property of one consistent subinstance being
preferred to another. Improvements are basically lifting of the
priority relation from facts to consistent subinstances; $J$ is an
improvement of $K$ if $J$ contains a fact that is better than all
those in $K\setminus J$ (in the Pareto semantics), or if for every
fact in $K\setminus J$ there exists a better fact in $J\setminus K$
(in the global semantics). In each of the two semantics, an \e{optimal
	repair} is a repair that cannot be improved. A third semantics
proposed by Staworko et al.~\cite{DBLP:journals/amai/StaworkoCM12} is
that of a \e{completion-optimal} repair, which is a globally optimal
repair under some extension of the priority relation into a \e{total}
relation. In this paper, we refer to these preferred repairs as
\e{p-repair}, \e{g-repair} and \e{c-repair}, respectively.

Fagin et al.~\cite{DBLP:conf/pods/FaginKRV14} have built on the concept of preferred
repairs (in conjunction with the framework of \e{document
	spanners}~\cite{DBLP:journals/jacm/FaginKRV15}) to devise a
language for declaring \e{inconsistency cleaning} in text
information-extraction systems. They have shown there that preferred
repairs capture ad-hoc cleaning operations and strategies of some
prominent existing systems for text
analytics~\cite{TIPSTER98,ChiticariuKLRRV10}. 

Staworko et al.~\cite{DBLP:journals/amai/StaworkoCM12} have proved
several results about preferred repairs. For example, every c-repair
is also a g-repair, and every g-repair is also a p-repair. They also
showed that p-repair and c-repair checking are solvable in polynomial
time (under data complexity) when constraints are given as denial
constraints, and that there is a set of functional dependencies (FDs)
for which g-repair checking is coNP-complete. Later, Fagin et
al.~\cite{DBLP:conf/pods/FaginKK15} extended that hardness result to a
full dichotomy in complexity over all sets of FDs: g-repair checking
is solvable in polynomial time whenever the set of FDs is equivalent
to a single FD or two key constraints per relation; in every other
case, the problem is coNP-complete.

While the classic complexity problems studied in the theory of repairs
include repair checking and consistent query answering, the presence
of repairs gives rise to the \e{cleaning problem}, which Staworko et
al.~\cite{DBLP:journals/amai/StaworkoCM12} refer to as
\e{categoricity}: determine whether the provided priority relation
suffices to clean the database unambiguously, or in other words,
decide whether there is exactly one optimal repair.  The problem of
repairing uniqueness (in a different repair semantics) is also
referred to as \e{determinism} by Fan et
al.~\cite{DBLP:journals/jdiq/FanM0Y14}. In this paper, we study the
three variants of this computational problem, under the three
optimality semantics Pareto, global and completion, and denote them as
\e{p-categoricity}, \e{g-categoricity} and \e{c-categoricity},
respectively.

It is known that under each of the three semantics there is always at
least one preferred repair, and Staworko et
al.~\cite{DBLP:journals/amai/StaworkoCM12} present a polynomial-time
algorithm for finding such a repair. (We recall this algorithm in
Section~\ref{sec:categoricity}.) Hence, the categoricity problem is
that of deciding whether the output of this algorithm is the only
possible preferred repair. As we explain next, it turns out that each
of the three variants of the problem entails quite a unique picture of
complexity.

For the problem of p-categoricity, we focus on integrity constraints
that are FDs, and establish the following dichotomy in data
complexity, assuming that $\mbox{P}\neq\mbox{NP}$.  For a relational
schema with a set $\Delta$ of FDs:
\begin{itemize}
	\item If $\Delta$ associates (up to equivalence) a single FD with
	every relation symbol, then p-categoricity is solvable in polynomial
	time.
	\item In \e{any other case}, p-categoricity is coNP-complete.
\end{itemize}
For example, with the relation symbol $R(A,B,C)$ and the FD
$A\rightarrow B$, p-categoricity is solvable in polynomial time; but
if we add the dependency $B\rightarrow A$ then it becomes
coNP-complete. Our proof uses a reduction technique from past
dichotomies that involve
FDs~\cite{DBLP:conf/pods/Kimelfeld12,DBLP:conf/pods/FaginKK15}, but
requires some highly nontrivial additions.

We then turn to investigating c-categoricity, and establish a far more
positive picture than the one for p-categoricity. In particular, the
problem is solvable in polynomial time for every set of FDs. In fact,
we present an algorithm for solving c-categoricity in polynomial time,
assuming that constraints are given as an input \e{conflict
	hypergraph}~\cite{DBLP:journals/iandc/ChomickiM05}. (In particular,
we establish polynomial-time data complexity for other types of
integrity constraints, such as \e{conditional
	FDs}~\cite{DBLP:conf/icde/BohannonFGJK07} and \e{denial
	constraints}~\cite{DBLP:journals/jiis/GaasterlandGM92}.)
The algorithm is extremely simple, yet its proof of correctness is
quite intricate.

Finally, we explore g-categoricity, and focus first on FDs. We show
that in the tractable case of p-categoricity (equivalence to a single
FD per relation), g-categoricity is likewise solvable in polynomial
time.  For example, $R(A,B,C)$ with the dependency $A\rightarrow B$
has polynomial-time g-categoricity.  Nevertheless, we prove that if
the we add the dependency $\emptyset\rightarrow C$ (that is, the
attribute $C$ should have the same value across all tuples), then
g-categoricity becomes $\piptwo$-complete.  We do not complete a
dichotomy as in p-categoricity, and leave that open for future work.
Lastly, we observe that in our proof of $\piptwo$-hardness, our
reduction constructs a non-transitive priority relation, and we ask
whether transitivity makes a difference. The three semantics of
repairs remain different in the presence of transitivity. In
particular, we show such a case where there are globally-optimal
repairs that are not completion optimal repairs. Nevertheless, quite
interestingly, we are able to prove that g-categoricity and
c-categoricity are actually \e{the same problem} if transitivity is
assumed. In particular, we establish that in the presence of
transitivity, g-categoricity is solvable in polynomial time, even when
constraints are given as a conflict hypergraph.

\eat{
	The rest of the paper is organized as follows. In
	Section~\ref{sec:preliminaries} we give preliminary definitions and
	notation.  We define the problem of categoricity in
	Section~\ref{sec:categoricity}.  In Section~\ref{sec:insights} we give
	some preliminary insights and recall relevant results from the
	literature. The complexity of the problems p-categoricity,
	c-categoricity and g-categoricity is investigated in
	Sections~\ref{sec:p}, \ref{sec:c} and~\ref{sec:g},
	respectively. Finally, we conclude and discuss future directions in
	Section~\ref{sec:conclusions}.
}

\section{Preliminaries}\label{sec:preliminaries}
We now present some general terminology and notation that we use
throughout the paper.

\subsection{Signatures and Instances}

A (\e{relational}) \e{signature} is a finite set
$\relset=\set{R_1,\dots,R_n}$ of \e{relation symbols}, each with a
designated positive integer as its \e{arity}, denoted
$\arity(R_i)$. We assume an infinite set $\consts$ of \e{constants},
used as database values. An \e{instance} $I$ over a signature
$\relset=\set{R_1,\dots,R_n}$ consists of finite relations
$R_i^I\subseteq\consts^{\arity(R_i)}$, where $R_i\in\relset$. We write
$\inds{R_i}$ to denote the set $\set{1,\dots,\arity(R_i)}$, and we
refer to the members of $\inds{R_i}$ as \e{attributes} of $R_i$.  If
$I$ is an instance over $\relset$ and $\tup t$ is a tuple in $R_i^I$,
then we say that $R_i(\tup t)$ is a \e{fact of $I$}. By a slight abuse
of notation, we identify an instance $I$ with the set of its
facts. For example, $R_i(\tup t)\in I$ denotes that $R_i(\tup t)$ is a
fact of $I$.  As another example, $J\subseteq I$ means that
$R_i^J\subseteq R_i^I$ for every $R_i\in\relset$; in this case, we say
that $J$ is \e{subinstance} of $I$.

\def\compceo{\rel{CompCEO}}
\def\comp{\att{company}}
\def\ceo{\att{ceo}}

In our examples, we often name the attributes and refer to them by
their names. For instance, in Figure~\ref{fig:companyceo-instance} we
refer to the relation symbol as $\compceo(\comp,\ceo)$ where $\comp$
and $\ceo$ refer to Attributes~1 and~2, respectively. In the case of
generic relation symbols, we implicitly name their attributes by
capital English letters with the corresponding numeric values; for
instance, we may refer to Attributes $1$, $2$ and $3$ of $R/3$ by $A$,
$B$ and $C$, respectively. We stress that attribute names are not part
of our formal model, but are rather used for readability.

\def\fgpi{f^{\val{g}}_{\val{pi}}}
\def\fgpa{f^{\val{g}}_{\val{pa}}}
\def\fgbr{f^{\val{g}}_{\val{br}}}
\def\fapa{f^{\val{a}}_{\val{pa}}}
\def\fapi{f^{\val{a}}_{\val{pi}}}

\subsection{Integrity and Inconsistency}
Let $\relset$ be a signature, and $I$ an instance over $\relset$. In
this paper we consider two representation systems for integrity
constraints. The first is \e{functional dependencies} and the second
is \e{conflict hypergraphs}.

\subsubsection{Functional Dependencies}
Let $\relset$ be a signature. A \e{Functional Dependency} (\e{FD} for
short) over $\relset$ is an expression of the form $R:X\rightarrow Y$,
where $R$ is a relation symbol of $\relset$, and $X$ and $Y$ are
subsets of $\inds R$. When $R$ is clear from the context, we may
omit it and write simply $X\rightarrow Y$.  A
special case of an FD is a \e{key constraint}, which is an FD of the
form $R:X\rightarrow Y$ where $X\cup Y=\inds R$. An FD
$R:X\rightarrow Y$ is \e{trivial} if $Y\subseteq X$; otherwise, it is
\e{nontrivial}.

When we are using the alphabetic attribute notation, we may write $X$
and $Y$ by simply concatenating the attribute symbols. For example, if
we have a relation symbol $R/3$, then $A\rightarrow BC$ denotes the FD
$R:\set{1}\rightarrow\set{2,3}$. An instance $I$ over $R$
\e{satisfies} an FD $R:X\rightarrow Y$ if for every two facts $f$ and
$g$ over $R$, if $f$ and $g$ agree on (i.e., have the same values for)
the attributes of $X$, then they also agree on the attributes of $Y$.
We say that $I$ satisfies a set $\Delta$ of FDs if $I$ satisfies every
FD in $\Delta$; otherwise, we say that $I$ \e{violates} $\Delta$. Two
sets $\Delta$ and $\Delta'$ of FDs are \e{equivalent} if for every
instance $I$ over $\relset$ it holds that $I$ satisfies $\Delta$ if
and only if it satisfies $\Delta'$. For example, for $R/3$ the sets
$\set{A\rightarrow BC,C\rightarrow A}$ and $\set{A\rightarrow
	C,C\rightarrow AB}$ are equivalent.

In this work, a \e{schema} $\scs$ is a pair $(\relset,\Delta)$, where
$\relset$ is a signature and $\Delta$ is a set of FDs over
$\relset$. If $\scs=(\relset,\depset)$ is a schema and $R\in\relset$,
then we denote by $\depset_{|R}$ the restriction of $\depset$ to the
FDs $R:X\rightarrow Y$ over $R$.

\begin{figure}[t]
	\centering
	\def\arraystretch{1.1}%
	\begin{tabular}{r|c|c|}
		\multicolumn{1}{c}{}& \multicolumn{2}{c}{$\compceo$} \\ \cline{2-3}
		& $\comp$ & $\ceo$ \\ \cline{2-3}
		$\fgpi$ & $\val{Google}$ &  $\val{Pichai}$ \\
		$\fgpa$ & $\val{Google}$  & $\val{Page}$ \\
		$\fgbr$ & $\val{Google}$  & $\val{Brin}$ \\
		$\fapa$ & $\val{Alphabet}$ & $\val{Page}$ \\
		$\fapi$ & $\val{Alphabet}$ & $\val{Pichai}$ \\ \cline{2-3}
	\end{tabular}
	\caption{Inconsistent database of the company-CEO running example}
	\label{fig:companyceo-instance}
\end{figure}

\begin{example}\label{example:ceo}
	In our first running example, we use the schema
	$\scs=(\relset,\Delta)$, defined as follows. The signature $\relset$
	consists of a single relation $\compceo(\comp,\ceo)$, which
	associates companies with their Chief Executive Officers
	(CEO). Figure~\ref{fig:companyceo-instance} depicts an instance $I$
	over $\relset$.  We define $\Delta$ as the following set of FDs over
	$\relset$.
	\[\Delta\eqdef\set{\comp\rightarrow\ceo\,,\, \ceo\rightarrow\comp}\]
	Hence, $\Delta$ states that in $\compceo$, each company has a single
	CEO and each CEO manages a single company. Observe that $I$ violates
	$\Delta$. For example, $\val{Google}$ has three CEOs, $\val{Alphabet}$
	has two CEOs, and each of $\val{Pichai}$ and $\val{Page}$ is the CEO
	of two companies.
\end{example}

\subsubsection{Conflict Hypergraphs}

While FDs define integrity logically, at the level of the signature, a
\e{conflict hypergraph}~\cite{DBLP:journals/iandc/ChomickiM05}
provides a direct specification of inconsistencies at the instance
level, by explicitly stating sets of tuples that cannot co-exist. In
the case of FDs, the conflict hypergraph is a graph that has an edge
between every two facts that violate an FD.  Formally, for an instance
$I$ over a signature $\relset$, a \e{conflict hypergraph} $\graph$
(\e{for $I$}) is a hypergraph that has the facts of $I$ as its node
set. A subinstance $J$ of $I$ is \e{consistent} with respect to
(w.r.t.)  $\graph$ if $J$ is an \e{independent set} of $\graph$; that
is, no hyperedge of $\graph$ is a subset of $J$.  We say that $J$ is
\e{maximal} if $J\cup\set{f}$ is inconsistent for every $f\in
I\setminus J$. When all the edges of a conflict hypergraph are of size
two, we may call it a \e{conflict graph}.

Recall that conflict hypergraphs can represent inconsistencies for
various types of integrity constraints, including FDs, the more
general \e{conditional FDs}~\cite{DBLP:conf/icde/BohannonFGJK07}, and
the more general \e{denial
	constraints}~\cite{DBLP:journals/jiis/GaasterlandGM92}.  In fact,
every constraint that is anti-monotonic (i.e., where subsets of
consistent sets are always consistent) can be represented as a
conflict hypergraph. In the case of denial constraints, the
translation from the logical constraints to the conflict hypergraph
can be done in polynomial time under \e{data complexity} (i.e., when
the signature and constraints are assumed to be fixed).

Let $\scs=(\relset,\Delta)$ be a schema, and let $I$ be an instance
over $\scs$. Recall that $\scs$ is assumed to have only FDs. We denote
by $\graph_{\scs}^I$ the conflict graph for $I$ that has an edge
between every two facts that violate some FD of $\scs$. Note that a
subinstance $J$ of $I$ satisfies $\Delta$ if and only if $J$ is
consistent w.r.t.~$\graph_{\scs}^I$. As an example, the left graph of
Figure~\ref{fig:ceo-completions} depicts the graph $\graph_{\scs}^I$
for our running example; for now, the reader should ignore the
directions on the edges, and view the graph as an undirected one. The
following example involves a conflict hypergraph that is not a graph.

\begin{example}\label{example:followers-instance}
	In our second running example, we use the toy scenario where the
	signature has a single relation symbol $\rel{Follows}/2$, where
	$\rel{Follows}(x,y)$ means that person $x$ follows person $y$ (e.g.,
	in a social network). We have two sets of people: $\val{a_i}$ for
	$i=1,2,3$, and $\val{b_j}$ for $j=1,\dots,5$. All the facts have the
	form $\rel{Follows}(\val{a_i},\val{b_j})$; we denote such a fact by
	$f_{ij}$. The instance $I$ has the following facts:
	\begin{center}$f_{11}$, $f_{12}$, $f_{21}$, $f_{22}$, $f_{23}$, $f_{24}$, $f_{31}$, $f_{32}$, $f_{34}$, $f_{35}$
	\end{center}
	The hypergraph $\graph$ for $I$ encodes the following rules:
	\begin{citemize}
		\item Each $\val{a_i}$ can follow at most $i$ people.
		\item Each $\val{b_j}$ can be followed by at most $j$ people.
	\end{citemize}
	Specifically, $\graph$ contains the following hyperedges:
	\begin{citemize}
		\item $\set{f_{11},f_{12}}$, $\set{f_{21},f_{22},f_{23}}$, $\set{f_{21},f_{22},f_{24}}$, \\$\set{f_{21},f_{23},f_{24}}$,
		$\set{f_{22},f_{23},f_{24}}$, $\set{f_{31},f_{32},f_{34},f_{35}}$
		\item $\set{f_{11},f_{21}}$, $\set{f_{11},f_{31}}$, $\set{f_{21},f_{31}}$, $\set{f_{12},f_{22},f_{32}}$
	\end{citemize}
	An example of a consistent subinstance $J$ is
	\[\set{f_{11},f_{22},f_{23},f_{32},f_{34},f_{35}}\,.\]
	The reader can verify that $J$ is maximal.
\end{example}

{
	\subsection{Prioritizing Inconsistent Databases}
	We now recall the framework of preferred repairs by Staworko et
	al.~\cite{DBLP:journals/amai/StaworkoCM12}. Let $I$ be an instance
	over a signature $\relset$.  A \e{priority} relation $\succ$ over $I$
	is an acyclic binary relation over the facts in $I$.  By \e{acyclic}
	we mean that $I$ does not contain any sequence $f_1,\dots,f_k$ of
	facts such that $f_i\succ f_{i+1}$ for all $i=1,\dots,k-1$ and
	$f_k\succ f_1$. If $\succ$ is a priority relation over $I$ and $K$ is
	a subinstance of $I$, then $\max_\succ(K)$ denotes the set of tuples
	$f\in K$ such that no $g\in K$ satisfies $g\succ f$.

	An \e{inconsistent prioritizing instance} over $\relset$ is a triple
	$(I,\graph,\succ)$, where $I$ is an instance over $\relset$, $\graph$
	is a conflict hypergraph for $I$, and $\succ$ a priority relation over
	$I$ with the following property: for every two facts $f$ and $g$ in
	$I$, if $f\succ g$ then $f$ and $g$ are neighbors in $\graph$ (that
	is, $f$ and $g$ co-occur in some hyperedge).\footnote{This requirement
		has been made with the introduction of the
		framework~\cite{DBLP:journals/amai/StaworkoCM12}. Obviously, the
		lower bounds we present hold even without this requirement.
		Moreover, our main upper bound,
		Theorem~\ref{thm:c-categoricity-ptime}, holds as well without this
		requirement. We defer to future work the thorough investigation of
		the impact of relaxing this requirement.}  For example, if
	$\H=\H_\scs^I$ (where all the constraints in $\scs$ are FDs), then
	$f\succ g$ implies that $\set{f,g}$ violates at least one FD.
	
	\begin{example}\label{example:ceo-priority}
		We continue our running company-CEO example. We define a priority
		relation $\succ$ by $\fgpi\succ \fgpa$, $\fgpa\succ \fgbr$ and
		$\fapa\succ \fapi$. We denote $\succ$ by corresponding arrows on the
		left graph of Figure~\ref{fig:ceo-completions}. (Therefore, some of the
		edges are directed and some are undirected.) We then get the
		inconsistent prioritizing instance $(I,\H_\scs^I,\succ)$ over
		$\relset$. Observe that the graph does not contain directed cycles,
		as required from a priority relation.
	\end{example}
	
	\begin{example}\label{example:followers-priority}
		Recall that the instance $I$ of our followers example is defined in
		Example~\ref{example:followers-instance}. The priority relation
		$\succ$ is given by $f_{il}\succ f_{jk}$ if one of the following
		holds: \e{(a)} $i=j$ and $k=l+1$, \e{or (b)} $j=i+1$ and $l=k$.  For
		example, we have $f_{11}\succ f_{12}$ and $f_{12}\succ f_{22}$. But
		we do not have $f_{11}\succ f_{22}$ (hence, $\succ$ is not
		transitive).
	\end{example}

	Let $(I,\graph,\succ)$ be an inconsistent prioritizing instance over a
	signature $\relset$. We say that $\succ$ is \e{total} if for every two
	facts $f$ and $g$ in $I$, if $f$ and $g$ are neighbors then either
	$f\succ g$ or $g\succ f$. A priority $\succ_c$ over $I$ is a
	\e{completion} of $\succ$ (w.r.t.~$\graph$) if $\succ$ is a subset of
	$\succ_c$ and $\succ_c$ is total. As an example, the middle and right
	graphs of Figure~\ref{fig:ceo-completions} are two completions of the
	priority relation $\succ$ depicted on the left side.  A \e{completion}
	of $(I,\graph,\succ)$ is an inconsistent prioritizing instance
	$(I,\graph,\succ_c)$ where $\succ_c$ is a completion of $\succ$.

	\begin{figure}[t]
		\centering
		\input{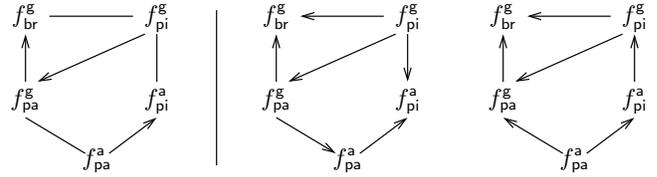}
		\caption{\label{fig:ceo-completions} The conflict graph $\graph_{\scs}^I$ and
		the priority relation $\succ$ for the company-CEO running example
		(left), and two completions of $\succ$ (middle and right)}
	\end{figure}

	\subsection{Preferred Repairs}
	
	Let $D=(I,\graph,\succ)$ be an inconsistent prioritizing instance over
	$\relset$. As defined by Arenas et
	al.~\cite{DBLP:conf/pods/ArenasBC99}, $J$ is a \e{repair} of $D$ if
	$J$ is a maximal consistent subinstance of $I$. Staworko et
	al.~\cite{DBLP:journals/amai/StaworkoCM12} define three different
	notions of \e{preferred} repairs: \e{Pareto optimal}, \e{globally
		optimal}, and \e{completion optimal}. The first two notions are
	based on checking whether a repair $J$ of $I$ can be improved by
	replacing a set of facts in $J$ with a more preferred set of facts
	from $I$. They differ by the way they define when one set of facts is
	considered more preferred than another one. The last notion is based
	on the notion of completion. Next we give the formal definitions.
	
	\begin{definitionthm}[Improvement]
		Let $(I,\graph,\succ)$ be an inconsistent prioritizing instance over
		a signature $\relset$, and $J$ and $J'$ two distinct consistent
		subinstances of $I$.
		\begin{itemize}
			\item $J$ is a \e{Pareto improvement} of $J'$ if there exists a fact
			$f\in J\setminus J'$ such that $f\succ f'$ for all facts $f'\in
			J'\setminus J$.
			\item $J$ is a \e{global improvement} of $J'$ if for every fact $f'\in
			J'\setminus J$ there exists a fact $f\in J\setminus J'$ such that
			$f\succ f'$.
		\end{itemize}
	\end{definitionthm}
	
	That is, $J$ is a Pareto improvement of $J'$ if, in order to obtain
	$J$ from $J'$, we insert and delete facts, and one of the inserted
	facts is preferred to all deleted facts. And $J$ is a global
	improvement of $J'$ if, in order to obtain $J$ from $J'$, we insert
	and delete facts, and every deleted fact is preferred to by some
	inserted fact.
	
	\begin{example}\label{example:ceo-improvement}
		We continue the company-CEO running example. We define three
		consistent subinstances of $I$.
		\begin{gather*}
			J_1\eqdef\set{\fgbr,\fapi}\quad J_2\eqdef\set{\fgpa,\fapi}\\
			J_3\eqdef\set{\fgbr,\fapa}\quad J_4\eqdef\set{\fgpi,\fapa}
		\end{gather*}
		Note the following. First, $J_2$ is a Pareto improvement of $J_1$,
		since $\fgpa\in J_2\setminus J_1$ and $\fgpa\succ f$ for every fact in
		$J_1\setminus J_2$ (where in this case there is only one such an
		$f$, namely $\fgbr$). Second, $J_4$ is a global improvement of $J_2$
		because $\fgpi \succ \fgpa$ and $\fapa\succ \fapi$. (We refer to $J_3$
		in later examples.)
	\end{example}
	
	We then get the following variants of \e{preferred repairs}.
	
	\begin{definitionthm}[p/g/c-repair]
		Let $D$ be an inconsistent prioritizing instance $(I,\graph,\succ)$,
		and let $J$ be a consistent subinstance of $I$. Then $J$ is a:
		\begin{itemize}
			\item \e{Pareto-optimal repair of $D$} if there is no Pareto
			improvement of $J$.
			\item \e{globally-optimal repair of $D$} if there is no global
			improvement of $J$.
			\item \e{completion-optimal repair of $D$} if there exists a completion
			$D_c$ of $D$ such that $J$ is a globally-optimal repair of $D_c$.
		\end{itemize}
		We abbreviate ``Pareto-optimal repair,'' ``globally-optimal repair,'' 
		and ``completion-optimal repair'' by \e{p-repair}, \e{g-repair} and
		\e{c-repair}, respectively.
	\end{definitionthm}
	
	We remark that in the definition of a completion-optimal repair, we
	could replace ``globally-optimal'' with ``Pareto-optimal'' and obtain
	an equivalent definition~\cite{DBLP:journals/amai/StaworkoCM12}.
	
	Let $D=(I,\graph,\succ)$ be an inconsistent prioritizing instance over
	a signature $\relset$. We denote the set of all the repairs,
	p-repairs, g-repairs and c-repairs of $D$ by $\rep{D}$, $\prep{D}$,
	$\grep{D}$ and $\crep{D}$, respectively. The following was shown by
	Staworko et al.~\cite{DBLP:journals/amai/StaworkoCM12}.
	
	\begin{citedproposition}{DBLP:journals/amai/StaworkoCM12}\label{prop:containments}
		For all inconsistent prioritizing instances $D$ we have
		$\crep{D}\neq\emptyset$, and moreover,
		\[\crep{D}\subseteq\grep{D}\subseteq \prep{D} \subseteq \rep{D}\,.\]
	\end{citedproposition}
	
	\begin{example}\label{example:ceo-repairs}
		We continue our company-CEO example. Recall the instances $J_i$
		defined in Example~\ref{example:ceo-improvement}. We have shown that
		$J_1$ has a Pareto improvement, and therefore, $J_1$ is \e{not} a
		p-repair (although it is a repair in the ordinary sense). The reader
		can verify that $J_2$ has no Pareto improvements, and therefore, it
		is a p-repair. But $J_2$ is not a g-repair, since $J_4$ is a global
		improvement of $J_2$. The reader can verify that $J_3$ is a g-repair
		(hence, a p-repair). Finally, observe that $J_4$ is a g-repair
		w.r.t.~the left completion of $\succ$ in
		Figure~\ref{fig:ceo-completions} (and also w.r.t.~the right
		one). Hence, $J_4$ is a c-repair (hence, a g-repair and a
		p-repair). In constrast, observe that $J_3$ has a global improvement
		(and a Pareto improvement) in both completions; but it does not
		prove that $J_3$ is not a c-repair (since, conceptually, one needs
		to consider all possible completions of $\succ$).
	\end{example}
	
	\begin{example}\label{example:follows-repair}
		We now continue the follower example. The inconsistent
		prioritizing instance $(I,\graph,\succ)$ is defined in
		Examples~\ref{example:followers-instance}
		and~\ref{example:followers-priority}. Consider the following
		instance.
		\[J_1\eqdef\set{f_{11},f_{22},f_{23},f_{32},f_{34},f_{35}}\] The
		reader can verify that $J_1$ is a c-repair (e.g., by completing
		$\succ$ through the lexicographic order). The subinstance
		$J_2=\set{f_{12},f_{21},f_{22},f_{34},f_{35}}$ is a repair but not a
		p-repair, since we can add $f_{11}$ and remove both $f_{12}$ and
		$f_{21}$, and thus obtain a Pareto improvement.
	\end{example}

}

\section{Categoricity}\label{sec:categoricity}
In this section we define the computational problem of
\e{categoricity}, which is the main problem that we study in this
paper.  Proposition~\ref{prop:containments} states that, under each of
the semantics of preferred repairs, at least one such a repair
exists. In general, there can be many possible preferred repairs. The
problem of \e{categoricity}~\cite{DBLP:journals/amai/StaworkoCM12} is
that of testing whether there is \e{precisely} one such a repair; that
is, there do not exist two distinct preferred repairs, and therefore,
the priority relation contains enough information to clean the
inconsistent instance unambiguously.
\begin{problem}
	Problems \e{p-catego\-ricity}, \e{g-catego\-ricity},\,\,  and
	\e{c-categoricity} are those of testing whether $|\prep{D}|=1$,
	$|\grep{D}|=1$ and $|\crep{D}|=1$, respectively, given a signature
	$\relset$ and an inconsistent prioritizing instance $D$ over
	$\relset$.
\end{problem}

As defined, categoricity takes as input both the signature $\relset$
and the inconsistent prioritizing instance $D$, where constraints are
represented by a conflict hypergraph. We also study this problem from
the perspective of \e{data complexity}, where we fix a schema
$\scs=(\relset,\Delta)$, where $\Delta$ is a set of FDs. In that case,
the input consists of an instance $I$ over $\relset$ and a priority
relation $\prec$ over $I$. The conflict hypergraph is then implicitly
assumed to be $\graph_{\scs}^I$. We denote the corresponding variants
of the problem by p-categoricity$\angs{\scs}$,
g-categoricity$\angs{\scs}$ and c-categoricity$\angs{\scs}$,
respectively.

\begin{example}\label{example:ceo-categoricity}
	Continuing our company-CEO example, we showed in
	Example~\ref{example:ceo-repairs} that there are at least two
	g-repairs and at least three p-repairs. Hence, a solver for
	g-categoricity$\angs{\scs}$ should return false on $(I,\succ)$, and
	so is a solver for p-categoricity$\angs{\scs}$. In contrast, we will
	later show that there is precisely one c-repair
	(Example~\ref{example:ccat-ceo}); hence, a solver for
	c-categoricity$\angs{\scs}$ should return true on $(I,\succ)$. If,
	on the other hand, we replaced $\succ$ with any of the completions
	in Figure~\ref{fig:ceo-completions}, then there would be precisely
	one p-repair and one g-repair (namely, the current single
	c-repair). This follows from a result of Staworko et
	al.~\cite{DBLP:journals/amai/StaworkoCM12}, stating that
	categoricity holds in the case of total priority relations.
\end{example}

\section{Preliminary Insights}\label{sec:insights}

We begin with some basic insights into the different variants of the
categoricity problem.

\subsection{Greedy Repair Generation}

We recall an algorithm by Staworko et
al.~\cite{DBLP:journals/amai/StaworkoCM12} for greedily constructing a
c-repair.  This is the algorithm $\algname{FindCRep}$ of
Figure~\ref{alg:ccat-opt-alg}. The algorithm takes as input an
inconsistent prioritizing instance $(I,\graph,\succ)$ and returns a
c-repair $J$.  It begins with an empty $J$, and incrementally inserts
tuples to $J$, as follows.  In each iteration of lines~3--6, the
algorithm selects a fact $f$ from $\max_\succ(I)$ and removes it from
$I$.  Then, $f$ is added to $J$ if it does not violate consistency,
that is, if $\graph$ does not contain any hyperedge $e$ such that
$e\subseteq J\cup\set{f}$. The specific way of choosing the fact $f$
among all those in $\max_\succ(I)$ is (deliberately) left unspecified,
and hence, different executions may result in different c-repairs. In
that sense, the algorithm is nondeterministic.  Staworko et
al.~\cite{DBLP:journals/amai/StaworkoCM12} proved that the possible
results of these different executions are \e{precisely} the c-repairs.

\begin{citedtheorem}{DBLP:journals/amai/StaworkoCM12}\label{thm:cgreedy}
	Let $(I,\graph,\succ)$ be an inconsistent prioritizing instance over
	$\relset$.  Let $J$ be a consistent subinstance of $I$. Then $J$ is
	a c-repair if and only if there exists an execution of
	$\algname{FindCRep}(I,\graph,\succ)$ that returns $J$.
\end{citedtheorem}

\begin{algseries}{t}{\label{alg:ccat-opt-alg}Finding a
		c-repair~\cite{DBLP:journals/amai/StaworkoCM12}}
	\begin{insidealg}{FindCRep}{$I,\graph,\succ$}
		\STATE $J:=\emptyset$
		\WHILE{$\max_\succ(I) \neq\emptyset$}
		\STATE choose a fact $f$ in $\max_\succ(I)$
		\STATE $I:=I\setminus \set{f}$
		\IF{$J\cup\set{f}$ is consistent w.r.t.~$\graph$}
		\STATE $J:=J\cup\set{f}$
		\ENDIF
		\ENDWHILE
		\STATE \textbf{return} $J$
	\end{insidealg}
\end{algseries}

Due to Theorem~\ref{thm:cgreedy}, we often refer to a c-repair as a
\e{greedy} repair. This theorem, combined with
Proposition~\ref{prop:containments}, has several implications for us.
First, we can obtain an x-repair (where x is either p, g or c) in
polynomial time. Hence, if a solver for x-categoricity determines that
there is a single x-repair, then we can actually generate that
x-repair in polynomial time. Second, c-categoricity is the problem of
testing whether $\algname{FindCRep}(I,\graph,\succ)$ returns the same
instance $J$ on every execution.  Moreover, due to
Proposition~\ref{prop:containments}, p-categoricity
(resp.~g-categoricity) is the problem of testing whether every
p-repair (resp.~g-repair) is equal to the one that is obtained by some
execution of the algorithm.

\begin{example}
	We consider the application of the algorithm $\algname{FindCRep}$ to
	the instance of our company-CEO example (where
	$\graph=\graph_{\scs}^I$). The following are two different
	executions. We denote inclusion in $J$ (i.e., the condition of
	line~5 is true) by plus and exclusion from $J$ by minus.
	\begin{citemize}
		\item $+\fgpi$, $-\fgpa$, $-\fgbr$, $+\fapa$, $-\fapi$.
		\item $+\fapa$, $-\fapi$, $+\fgpi$, $-\fgpa$, $-\fgbr$.
	\end{citemize}
	Observe that both executions return $J_4=\set{\fgpi,\fapa}$. This is in
	par with the statement in Example~\ref{example:ceo-categoricity}
	that in this running example there is a single c-repair.
\end{example}

\subsection{Complexity Insights}
Our goal is to study the complexity of x-categoricity (where x is g, p
and c). This problem is related to that of \e{x-repair checking},
namely, given $D=(I,\graph,\succ)$ and $J$, determine whether $J$ is
an x-repair of $D$. The following is known about this problem.

\begin{citedtheorem}{DBLP:journals/amai/StaworkoCM12,DBLP:conf/pods/FaginKK15}\label{thm:repairchecking}
	The following hold.
	\begin{itemize}
		\item p-repair checking and c-repair checking are solvable in
		polynomial time; g-repair checking is in
		coNP~\cite{DBLP:journals/amai/StaworkoCM12}.
		\item Let $\scs=(\relset,\depset)$ be a fixed schema.  If
		$\depset_{|R}$ is equivalent to either a single FD or two key
		constraints for every $R\in\relset$, then g-repair checking is
		solvable in polynomial time; otherwise, g-repair checking is coNP-complete~\cite{DBLP:conf/pods/FaginKK15}.
	\end{itemize}
\end{citedtheorem}

Recall from Proposition~\ref{prop:containments} that there is always
at least one x-repair. Therefore, given $(I,\graph,\succ)$ we can
solve the problem using a coNP algorithm with an oracle to x-repair
checking: for all two distinct subinstances $J_1$ and $J_2$, either
$J_1$ or $J_2$ is not an x-repair. Therefore, from
Theorem~\ref{thm:repairchecking} we conclude the following.

\begin{corollary}\label{cor:upperbounds}
	The following hold.
	\begin{itemize}
		\item p-categoricity and c-categoricity are in coNP.
		\item For all fixed schemas $\scs=(\relset,\depset)$,
		g-categoricity$\angs{\scs}$ is in $\piptwo$, and moreover, if
		$\depset_{|R}$ is equivalent to either a single FD or two key
		constraints for every $R\in\relset$ then g-categoricity$\angs{\scs}$
		is in coNP.
	\end{itemize}
\end{corollary}

We stress here that if x-categoricity is solvable in polynomial time,
then x-categoricity$\angs{\scs}$ is solvable in polynomial time for
\e{all} schemas $\scs$; this is true since for every fixed schema
$\scs$ the hypergraph $\graph_{\scs}^I$ can be constructed in
polynomial time, given $I$. Similarly, if x-categoricity$\angs{\scs}$
is coNP-hard (resp.~$\piptwo$-hard) for \e{at least one} $\scs$, then
x-categoricity is coNP-hard (resp.~$\piptwo$-hard).

When we are considering x-categoricity$\angs{\scs}$, we assume that
all the integrity constraints are FDs. Therefore, unlike the general
problem of x-categoricity, in x-categoricity$\angs{\scs}$ conflicting
facts always belong to the same relation. It thus follows that our
analysis for x-categoricity$\angs{\scs}$ can restrict to
single-relation schemas. Formally, we have the following.

\begin{proposition}\label{prop:single-relation}
	Let $\scs=(\relset,\depset)$ be a schema and x be one of 
	p, g and c. For each relation $R\in\relset$, let 
	$\scs_{|R}$ be the schema $(\set{R},\depset_{|R})$.
	\begin{citemize}
		\item If x-categoricity$\angs{\scs_{|R}}$ is solvable in polynomial
		time for every $R\in\relset$, then x-categoricity$\angs{\scs}$ is
		solvable in polynomial time.
		\item If x-categoricity$\angs{\scs_{|R}}$ is coNP-hard
		(resp.~$\piptwo$-hard) for at least one $R\in\relset$, then
		x-categoricity$\angs{\scs}$ is coNP-hard (resp.~$\piptwo$-hard).
	\end{citemize}
\end{proposition}

Observe that the phenomenon of Proposition~\ref{prop:single-relation}
\e{does not} hold for x-categoricity, since the given conflict
hypergraph may include hyperedges that cross relations.

In the following sections we investigate each of the three variants of
categoricity: 
p-categoricity (Section~\ref{sec:p}), 
c-categoricity (Section~\ref{sec:c}) and
g-categoricity (Section~\ref{sec:g}).

\newcommand{\pc}{p-categoricity$\angs{\scs}$\xspace}

\section{p-Categoricity}\label{sec:p}
In this section we prove a dichotomy in the complexity of \pc over all
schemas $\scs$ (where $\depset$ consists of FDs). This dichotomy
states that the only tractable case is where the schema associates a
single FD (which can be trivial) to each relation symbol, up to
equivalence.  In all other cases, \pc is coNP-complete. Formally, we
prove the following.

\begin{theorem}
	\label{thm:pareto}
	Let $\scs=(\relset,\Delta)$ be a schema. The problem \pc can be solved
	in polynomial time if $\Delta|_R$ is equivalent to a single {FD} for
	every $R\in \relset$.  In every other case, \pc is coNP-complete.
\end{theorem}

The proof of Theorem~\ref{thm:pareto} is involved, and we outline it
in the rest of this section. The tractability side is fairly simple
(as we show in the next section), and the challenge is in the hardness
side.  Due to Proposition \ref{prop:single-relation}, it suffices to
consider schemas $\scs$ with a single relation. Hence, in the
remainder of this section we consider only such schemas $\scs$.

\subsection{Proof of Tractability}
\label{pareto: tractability}
In this section we fix a schema $\scs = (\relset, \depset)$, such that
$\relset$ consist of a single relational symbol $R$.  We will prove
that \pc is solvable in polynomial time if $\depset$ is a singleton.
We denote the single FD in $\depset$ as $X\rightarrow Y$. We fix the
input $(I,\succ)$ for \pc.

For a fact $f\in R^I$, we denote by $f[X]$ and $f[Y]$ the restriction
of the tuple of $f$ to the attributes in $X$ and $Y$, respectively.
Adopting the terminology of Koutris and
Wijsen~\cite{DBLP:conf/pods/KoutrisW15}, a \e{block} of $I$ is a
maximal collection of facts of $I$ that agree on all the attributes of
$X$ (i.e., facts $f$ that have the same $f[X]$).  Similarly, a
\e{subblock} of $I$ is a maximal collection of facts that agree on
both $X$ and $Y$.  For tuples $\tup a$ and $\tup b$ constants, we
denote by $I_{\tup a}$ the block of facts $f$ with $f[X]=\tup a$, and
by $I_{\tup a,\tup b}$ the subblock of facts $f$ with $f[X]=\tup a$
and $f[Y]=\tup b$. 

\begin{example}
	Consider again the instance $I$ of
	Figure~\ref{fig:companyceo-instance}, and suppose that $\depset$
	consists of only $\att{company}\rightarrow\att{ceo}$ (i.e., each
	company has a single CEO, but a person can be the CEO of several
	companies). Then for $\tup a=(\val{Google})$ and $\tup
	b=(\val{Pichai})$ the block $I_{\tup a}$ is
	$\set{\fgpi,\fgpa,\fgbr}$ and the subblock $I_{\tup a,\tup b}$ is
	the singleton $\set{\fgpi}$.
\end{example}

Tractability for $\scs$ is based on the following
lemma.

\def\lemmakeyforparetoonefd{
	Let $J$ be a subinstance of $I$. Then $J$ is a p-repair if and only if
	$J$ is a union of p-repairs over all the blocks $I_{\tup a}$ of $I$.
	Moreover, each p-repair of a block $I_{\tup a}$ is a subblock
	$I_{\tup a,\tup b}$.
}

\begin{lemma}\label{lemma:key-for-pareto-1fd}
	\lemmakeyforparetoonefd
\end{lemma}

We then get the following lemma.
\def\lemmasinglefactorized
{ The following are equivalent.
	\begin{compactenum}
		\item $I$ has a single p-repair.
		\item Each block $I_{\tup a}$ has a single p-repair.
		\item No block $I_{\tup a}$ has two distinct subblocks $I_{\tup a,\tup
			b}$ and $I_{\tup a,\tup c}$ that are p-repairs of $I_{\tup a}$.
	\end{compactenum}
}

\begin{lemma}\label{lemma:single-factorized}
	\lemmasinglefactorized
\end{lemma}

A polynomial-time algorithm then follows directly from
Lemma~\ref{lemma:single-factorized} and the fact that p-repair
checking is solvable in polynomial time
(Theorem~\ref{thm:repairchecking}).

\subsection{Proof of Hardness}
The hardness side of the dichotomy is more involved than its
tractability side. Our proof is based on the concept of a \e{fact-wise
	reduction}~\cite{DBLP:journals/tods/KimelfeldVW12}, which has also
been used by Fagin et al.~\cite{DBLP:conf/pods/FaginKK15} in the
context of g-repair checking.

\subsubsection{Fact-Wise Reduction}

Let $\scs=(\relset,\depset)$ and $\scs'=(\relset',\depset')$ be two
schemas. A \e{mapping} from $\relset$ to $\relset'$ is a function
$\mu$ that maps facts over $\relset$ to facts over $\relset'$. We
naturally extend a mapping $\mu$ to map instances $J$ over $\relset$
to instances over $\relset'$ by defining $\mu(J)$ to be
$\set{\mu(f)\mid f\in J}$.  A \e{fact-wise reduction} from $\scs$ to
$\scs'$ is a mapping $\Pi$ from $\relset$ to $\relset'$ with the
following properties.
\begin{enumerate}
	\item $\Pi$ is injective; that is, for all facts $f$ and $g$ over
	$\relset$, if $\Pi(f) = \Pi(g)$ then $f = g$.
	\item $\Pi$ preserves consistency and inconsistency; that is, for
	every instance $J$ over $\scs$, the instance $\Pi(J)$ satisfies
	$\depset'$ if and only if $J$ satisfies $\depset$.
	\item $\Pi$ is computable in polynomial time.
\end{enumerate}

Let $\scs$ and $\scs'$ be two schemas, and let $\Pi$ be a fact-wise
reduction from $\scs$ to $\scs'$.  Given an inconsistent instance $I$
over $\scs$ and a priority relation $\succ$ over $I$, we denote by
$\Pi(\succ)$ the priority relation $\succ'$ over $\Pi(I)$ where
$\Pi(f)\succ'\Pi(g)$ if and only if $f\succ g$. If $D$ is the
inconsistent prioritizing instance $(I,\graph_{\scs}^I,\succ)$, then
we denote by $\Pi(D)$ the triple
$(\Pi(I),\graph_{\scs'}^{\Pi(I)},\Pi(\succ))$, which is also an
inconsistent prioritizing instance. The usefulness of fact-wise
reductions is due to the following proposition, which is
straightforward.

\begin{proposition}
	Let $\scs$ and $\scs'$ be two schemas, and suppose that $\Pi$ is a
	fact-wise reduction from $\scs$ to $\scs'$.  Let $I$ be an
	inconsistent instance over $\scs$, $\succ$ a priority relation over
	$I$, and $D$ the inconsistent prioritizing instance
	$(I,\graph_{\scs}^I,\succ)$.  Then there is a bijection between
	$\prep{D}$ and $\prep{\Pi(D)}$.
\end{proposition}

We then conclude the following corollary.
\begin{corollary}\label{cor:fact-wise}
	If there is a fact-wise reduction from $\scs$ to $\scs'$, then
	there is a polynomial-time reduction from p-categoricity$\angs{\scs}$ to
	p-categoricity$\angs{\scs'}$.
\end{corollary}

\subsubsection{Specific Schemas}\label{sec:spec-schemas}
In the proof we consider seven specific schemas. The importance of
these schemas will later become apparent.  We denote these schemas by
$\scs^i$, for $i=0,1,\dots,6$, where each $\scs^i$ is the schema
$(\relset^i,\depset^i)$, and $\relset^i$ is the singleton
$\set{R^i}$. The specification of the $\scs^i$ is as follows.

\begin{enumerate}
	\item[0.] $R^0/2$ and $\depset^0=\set{A\rightarrow B,B\rightarrow A}$
	\item[1.] $R^1/3$ and $\depset^1=\set{AB\rightarrow C,BC\rightarrow A,AC \rightarrow B}$
	\item[2.] $R^2/3$ and $\depset^2=\set{A\rightarrow B, B\rightarrow A}$ 
	\item[3.] $R^3/3$ and $\depset^3=\set{AB\rightarrow C, C\rightarrow B}$
	\item[4.] $R^4/3$ and $\depset^4=\set{A\rightarrow B, B\rightarrow C}$
	\item[5.] $R^5/3$ and $\depset^5=\set{A\rightarrow C , B\rightarrow C}$
	\item[6.] $R^6/3$ and $\depset^6=\set{\emptyset\rightarrow A, B\rightarrow C}$
\end{enumerate} 

(In the definition of $\scs^6$, recall that $\emptyset\rightarrow A$
denotes the FD $\emptyset\rightarrow\set{1}$, meaning that all tuples
should have the same value for their first attribute.) In the proof we
use fact-wise reductions from the $\scs^i$, as we explain in the next
section.

\subsubsection{Two Hard Schemas}
Our proof boils down to proving coNP-hardness for two specific
schemas, namely $\scs^0$ and $\scs^6$, and then using (known and new)
fact-wise reductions in order to cover all the other schemas.  For
$\scs^6$ the proof is fairly simple. But hardness for $\scs^0$ turns
out to be quite challenging to prove, and in fact, this part is the
hardest in the proof of Theorem~\ref{thm:pareto}.  Note that $\scs^0$
is the schema of our company-CEO running example (introduced in
Example~\ref{example:ceo}).

\def\theoremhardnessspecificpareto{
	The problems p-categoricity$\angs{\scs^0}$ and
	p-categoricity$\angs{\scs^6}$ are both coNP-hard.}

\begin{theorem}\label{thm:hardness-specific-pareto} 
	\theoremhardnessspecificpareto
\end{theorem}

The proof (as well as all the other proofs for the results in this
paper) can be found in the appendix.

	\begin{figure}[t]
		\centering
		\input{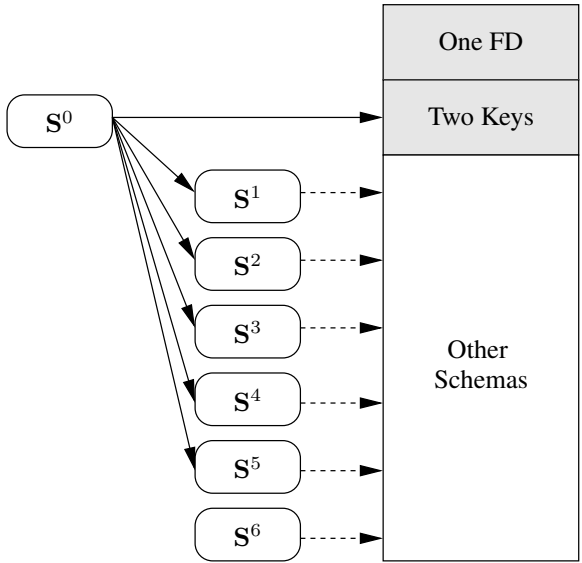}
		\caption{\label{fig:pcategoricitystrategy}The structure of fact-wise
	reductions for proving the hardness side of the dichotomy of Theorem~\ref{thm:pareto}}
	\end{figure}

\eat{
	Our proof of coNP-hardness
	of p-categoricity$\angs{\scs^0}$ is by a reduction from the
	\e{Exact-Cover problem}, and is quite intricate. For $\scs^6$ the
	proof is easier. Both proofs are in the appendix.
}

\subsubsection{Applying Fact-Wise Reductions}

The following has been proved by Fagin et
al.~\cite{DBLP:conf/pods/FaginKK15}.
\begin{citedtheorem}{DBLP:conf/pods/FaginKK15}\label{lemma:fw-from-s1-6}
	Let $\scs=(\relset,\depset)$ be a schema such that $\relset$
	consists of a single relation symbol. Suppose that $\depset$ is
	equivalent to neither any single FD nor any pair of keys. Then there
	is a fact-wise reduction from some $\scs^i$ to $\scs$, where
	$i\in\set{1,\dots,6}$.
\end{citedtheorem}

In the appendix we prove the following two lemmas, giving additional
fact-wise reductions.
\def\lemmafwfromszero
{Let $\scs=(\relset,\depset)$ be a schema such that $\relset$
	consists of a single relation symbol. Suppose that $\depset$ is
	equivalent to a pair of keys, and $\depset$ is not equivalent to any
	single FD. Then there is a fact-wise reduction from $\scs^0$ to
	$\scs$.}
\begin{lemma}\label{lemma:fw-from-s0}
	\lemmafwfromszero
\end{lemma}

\def\lemmafwfromszerotosonetofive
{  For all $i=1,\dots,5$ there is a fact-wise reduction from $\scs^0$
	to $\scs^i$.}
\begin{lemma}\label{lemma:fw-from-s0-to-s1-5}
	\lemmafwfromszerotosonetofive
\end{lemma}

The structure of our fact-wise reductions is depicted in
Figure~\ref{fig:pcategoricitystrategy}. Dashed edges are known
fact-wise reductions, while solid edges are novel. Observe that each
single-relation schema on the hardness side of
Theorem~\ref{thm:pareto} has an ingoing path from either $\scs^0$ or
$\scs^6$, both shown to have coNP-hard p-categoricity
(Theorem~\ref{thm:hardness-specific-pareto}).

\section{c-Categoricity}\label{sec:c}
We now investigate the complexity of c-categoricity. Our main result
is that this problem is tractable.

\begin{theorem}\label{thm:c-categoricity-ptime}
	The c-categoricity problem  is solvable in polynomial time.
\end{theorem}

In the remainder of this section we establish
Theorem~\ref{thm:c-categoricity-ptime} by presenting a polynomial-time
algorithm for solving c-categoricity.  The algorithm is very simple,
but its proof of correctness (given in the appendix) is intricate.

\subsection{Notation}
To present our algorithm, some notation is required.  Let
$(I,\graph,\succ)$ be an inconsistent prioritizing instance.  The
\e{transitive closure} of $\succ$, denoted $\succs$, is the priority
relation over the facts of $I$ where for every two facts $f$ and $g$ in
$I$ it holds that $f\succs g$ if and only if there exists a sequence
$f_0,\dots,f_m$ of facts, where $m>0$, such that $f=f_0$, $f_m=g$, and
$f_i\succ f_{i+1}$ for all $i=0,\dots,m-1$. Obviously, $\succs$ is
acyclic (since $\succ$ is acyclic). But unlike $\succ$, the relation
$\succs$ may compare between facts that are not necessarily neighbors
in $\graph$.

Let $(I,\graph,\succ)$ be an inconsistent prioritizing instance, let
$K$ be a set of facts of $I$, and let $f$ be a fact of $I$. By
$K\succs f$ we denote the fact that $g\succs f$ for \e{every} fact
$g\in K$.

\subsection{Algorithm}
\begin{algseries}{t}{\label{alg:ccat-alg}Algorithm for c-categoricity}
	\begin{insidealg}{CCategoricity}{$I,\graph,\succ$}
		\STATE $i:=0$
		\STATE $J:=\emptyset$
		\WHILE{$I \neq\emptyset$}
		\STATE $i:=i+1$
		\STATE $P_i:= \max_{\succssubs}(I)$
		\STATE $J:= J\cup P_i$
		\STATE $N_i:=\{f\in I\mid  \mbox{$\graph$ has a hyperedge $e$ s.t.~$f\in e$,}$\\ $\mbox{$\left(e\setminus\set{f}\right)\subseteq
			J$, and $(e \setminus \set{f})\succs f$}\}$
		\STATE $I:= I\setminus(P_i\cup N_i)$
		\ENDWHILE
		\STATE \textbf{return} true iff $J$ is consistent
	\end{insidealg}
\end{algseries}

Figure~\ref{alg:ccat-alg} depicts a polynomial-time algorithm for
solving c-categoricity. We next explain how it works, and later
discuss its correctness.

As required, the input for the algorithm is an inconsistent
prioritizing instance $(I,\graph,\succ)$. (The signature $\relset$ is
not needed by the algorithm.)  The algorithm incrementally constructs
a subinstance $J$ of $I$, starting with an empty $J$. Later we will
prove that there is a single c-repair if and only if $J$ is
consistent; and in that case, $J$ is the single c-repair.  The loop in
the algorithm constructs fact sets $P_1,\dots,P_t$ and
$N_1,\dots,N_t$. Each $P_i$ is called a \e{positive stratum} and each
$N_i$ is called a \e{negative stratum}.  Both $P_i$ and $N_i$ are
constructed in the $i$th iteration. On that iteration we add all the
facts of $P_i$ to $J$ and remove from $I$ all the facts of $P_i$ and
all the facts of $N_i$. The sets $P_i$ and $N_i$ are defined as
follows.
\begin{citemize}
	\item $P_i$ consists of the maximal facts in the current $I$,
	according to $\succs$.
	\item $N_i$ consists of all the facts $f$ that, together with
	$P_1 \cup\dots\cup P_i$, complete a hyperedge of preferred facts;
	that is, $\graph$ contains a hyperedge that contains $f$, is
	contained in $P_1\cup\dots\cup P_i\cup \set{f}$, and satisfies
	$g\succs f$ for every incident $g\neq f$.
\end{citemize}
The algorithm continues to iterate until $I$ gets empty.  As said
above, in the end the algorithm returns true if $J$ is consistent, and
otherwise false. Next, we give some examples of executions of the
algorithm.

\begin{example}\label{example:ccat-ceo}
	Consider the inconsistent prioritizing instance $(I,\graph,\succ)$
	from our company-CEO running example, illustrated on the left side
	of Figure~\ref{fig:ceo-completions}. The algorithm makes a single
	iteration on this instance, where $P_1=\set{\fgpi,\fapa}$ and
	$N_1=\set{\fgpa,\fapi,\fgpa}$. Both $\fgpi$ and $\fapa$ are in $P_1$
	since both are maximal. Also, each of $\fgpa$, $\fapi$ and $\fgpa$
	is in conflict with $P_1$, and we have $\fgpi\succ\fgpa$,
	$\fapa\succ\fapi$, and $\fgpi\succs\fgbr$.
\end{example}

\begin{example}
	Now consider the inconsistent prioritizing instance
	$(I,\graph,\succ)$ from our followers running
	example. Figure~\ref{fig:follows-exec} illustrates the execution of
	the algorithm, where each column describes $P_i$ or $N_i$, from left
	to right in the order of their construction. For convenience, the
	priority relation $\succ$, as defined in
	Example~\ref{example:followers-priority}, is depicted in
	Figure~\ref{fig:follows-exec} using corresponding edges between the
	facts.
	
	On iteration 1, for instance, we have $P_1=\set{f_{11},f_{34}}$,
	since $f_{11}$ and $f_{34}$ are the facts without incoming edges on
	Figure~\ref{fig:follows-exec}. Moreover, we have
	$N_1=\set{f_{12},f_{21},f_{31}}$.  The reason why $N_1$ contains
	$f_{12}$, for example, is that $\set{f_{11},f_{12}}$ is a
	hyperedge, the fact $f_{11}$ is in $P_1$, and $f_{11}\succ f_{12}$
	(hence, $f_{11}\succs f_{12}$). For a similar reason $N_1$ contains
	$f_{21}$. Fact $f_{31}$ is in $N_1$ as $\set{f_{11},f_{31}}$
	is a hyperedge, and though $f_{11}\not\succ f_{31}$, we have
	$f_{11}\succs f_{31}$. As another example, $N_3$ contains $f_{24}$
	since $\graph$ has the hyperedge $\set{f_{22},f_{23},f_{24}}$,
	the set $\set{f_{22},f_{23}}$ is contained in $P_1\cup P_2\cup P_3$,
	and $\set{f_{22},f_{23}}\succs f_{24}$.
	
	In the end, $J=\set{f_{11},f_{22},f_{23},f_{32},f_{34},f_{35}}$,
	which is also the subinstance $J_1$ of
	Example~\ref{example:follows-repair}. Since $J$ is consistent, the
	algorithm will determine that there is a single c-repair, and that
	c-repair is $J$.
\end{example}

	\begin{figure}[t]
		\centering
		\input{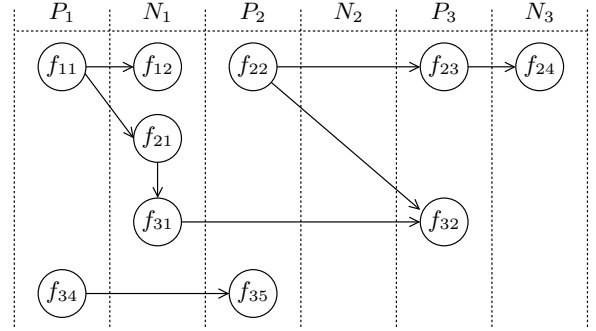}
		\caption{\label{fig:follows-exec}Execution of \algname{CCategoricity} on the
	followers example}
	\end{figure}

\begin{example}\label{example:ccat-fail}
	We now give an example of an execution on a negative instance of
	c-categoricity. (In Section~\ref{sec:g} we refer to this example for
	a different reason.)  Figure~\ref{fig:transitive-example} shows an
	instance $I$ over the schema $\scs^1$, which is defined in
	Section~\ref{sec:spec-schemas}.  Recall that in this schema every
	two attributes form a key. Each fact $R^1(a_1,a_2,a_3)$ in $I$ is
	depicted by a tuple that consists of the three values. For example,
	$I$ contains the (conflicting) facts $R^1(\val{A},\val{a},\val{1})$
	and $R^1(\val{A},\val{a},\val{2})$. Hereon, we write $\val{Xyi}$
	instead of $R^1(\val{X},\val{y},\val{i})$.  The priority relation
	$\succ$ is given by the directed edges between the facts; for
	example, $\val{Aa1}\succ\val{Aa2}$. Undirected edges are between
	conflicting facts that are incomparable by $\succ$ (e.g.,
	$\val{Ab2}$ and $\val{Ab3}$).  
	
	The execution of the algorithm on $(I,\graph_{\scs_1}^I,\succ)$ is
	as follows.  On the first iteration,
	$P_1=\set{\val{Aa1},\val{Ab2},\val{Ba3},\val{Bb1}}$ and
	$N_1=\set{\val{Aa2},\val{Bb3}}$.  In particular, note that $N_1$
	does not contain $\val{Ba2}$ since it conflicts only with
	$\val{Ba3}$ in $P_1$, but the two are incomparable.  Similarly,
	$N_1$ does not contain $\val{Ab3}$ since it is incomparable with
	$\val{Ab2}$. Consequently, in the second iteration we have
	$P_2=\set{\val{Ba2},\val{Ab3}}$ and $N_2=\emptyset$.  In the end,
	$J=P_1\cup P_2$ is inconsistent, and therefore, the algorithm will
	return false.  Indeed, the reader can easily verify that each of the
	following is a c-repair:
	$\set{\val{Aa1},\val{Ab2},\val{Ba3},\val{Bb1}}$,
	$\set{\val{Aa1},\val{Ab2},\val{Ba2},\val{Bb1}}$, and
	$\set{\val{Aa1},\val{Ba3},\val{Ab3},\val{Bb1}}$.
\end{example}

	\begin{figure}[t]
		\centering
		\input{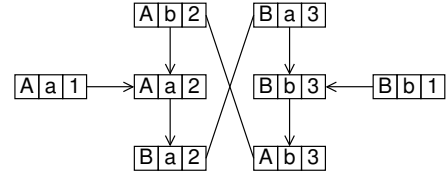}
		\caption{\label{fig:transitive-example}An inconsistent instance $I$ over
	$\scs^1$ with a priority relation $\succ$ over $I$}
	\end{figure}

\subsection{Correctness}
Correctness of $\algname{CCategoricity}$ is stated in the following
theorem.

\def\theoremccategoricitycorrect{
	Let $(I,\graph,\succ)$ be an inconsistent prioritizing instance, and
	let $J$ be the subinstance of $I$ constructed in the execution of
	$\algname{CCategoricity}(I,\graph,\succ)$. Then the following are
	equivalent.
	\begin{cenumerate}
		\item $J$ is consistent.
		\item There is a single c-repair.
	\end{cenumerate}
	Moreover, if $J$ is consistent then $J$ is the single c-repair.}

\begin{theorem}\label{thm:ccategoricity-correct}
	\theoremccategoricitycorrect
\end{theorem}

Theorem~\ref{thm:ccategoricity-correct}, combined with the observation
that the algorithm $\algname{CCategoricity}$ terminates in polynomial
time, imply Theorem~\ref{thm:c-categoricity-ptime}. As previously
said, the proof of Theorem~\ref{thm:ccategoricity-correct} is quite
involved. The direction $1\rightarrow 2$ is that of \e{soundness}---if
the algorithm returns true then there is precisely one c-repair. The
direction $2\rightarrow 1$ is that of \e{completeness}---if there is
precisely one c-repair then the algorithm returns true.

Soundness is the easier direction to prove. We assume, by way of
contradiction, that there is a c-repair $J'$ different from the
subinstance $J$ returned by the algorithm. Such $J'$ must include a
fact $f'$ from some negative stratum. We consider an execution of the
algorithm $\algname{FindCRep}$ that returns $J'$, and establish a
contradiction by considering the first time such an $f'$ is being
added to the constructed solution.

Proving completeness is more involved. We assume, by way of
contradiction, that the constructed $J$ is inconsistent. We are
looking at the first positive stratum $P_i$ such that
$P_1\cup\dots\cup P_i$ contains a hyperedge. Then, the crux of the
proof is in showing that we can then construct two c-repairs using the
algorithm $\algname{FindCRep}$: one contains some fact from $P_i$ and
another one does not contain that fact. We then establish that there
are at least two c-repairs, hence a contradiction.
\newcommand{\gc}{g-categoricity$\angs{\scs}$\xspace}

\section{g-Categoricity}\label{sec:g}
In this section, we investigate the complexity of the g-categoricity.
We first show a tractability result for the case of a schema with a
single FD. Then, we show $\piptwo$-completeness for a specific schema.
Finally, we discuss the implication of assuming \e{transitivity} in
the priority relation, and show a general positive result therein.

\subsection{Tractable Schemas}
Recall from Theorem~\ref{thm:pareto} that, assuming
$\mbox{P}\neq\mbox{NP}$, the problem p-categoricity$\angs{\scs}$ is
solvable in polynomial time if and only if $\scs$ consists (up to
equivalence) of a single FD per relation. The reader can verify that
the same proof works for g-categoricity$\angs{\scs}$. Hence, our first
result is that the tractable schemas of p-categoricity remain
tractable for g-categoricity.

\def\thmglobalonefd
{  Let $\scs=(\relset,\Delta)$ be a schema. The problem \gc can be
	solved in polynomial time if $\Delta|_R$ is equivalent to a single
	{FD} for every $R\in \relset$.
}
\begin{theorem}\label{thm:global-1fd}
	\thmglobalonefd
\end{theorem}

It is left open whether there is any schema $\scs$ that is not as in
Theorem~\ref{thm:global-1fd} where \gc is solvable in polynomial time.
In the next section we give an insight into this open problem (Theorem~\ref{thm:gcat-s0-hypothetical}).

\subsection{Intractable Schemas}

Our next result shows that \gc hits a harder complexity class than
\pc. In particular, while \pc is always in coNP (due to
Theorem~\ref{cor:upperbounds}), we will show a schema $\scs$ where \gc
is $\piptwo$-complete. This schema is the schema $\scs^6$ from
Section~\ref{sec:spec-schemas}.
\def\thmpiptwo{g-categoricity$\angs{\scs^6}$ is $\piptwo$-complete.}
\begin{theorem}\label{thm:piptwo}
	\thmpiptwo
\end{theorem}
The proof of Theorem~\ref{thm:piptwo} is by a reduction from the
$\piptwo$-complete problem $\qsattwo$: Given a CNF formula $\psi(\tup
x,\tup y)$, determine whether it is the case that for every truth
assignment to $\tup x$ there exists a truth assignment to $\tup y$
such that the two assignments satisfy $\psi$.

We can generalize Theorem~\ref{thm:piptwo} to a broad set of schemas,
by using fact-wise reductions from $\scs^6$. This is done in the
following theorem.

\eat{ Applying a fact-wise reduction from $\scs^6$, we can extend our
	result to a more general set of schemas. That is, we show that \gc
	is $\piptwo$-complete for other schemas rather than for $\scs^6$
	itself.  For instance, let $\scs=(\relset,\depset)$ where $\relset$
	consists of a single relation $R/4$ and
	$\depset = \{A\rightarrow B, C\rightarrow D\}$. The
	g-categoricity$\angs{\scs}$ is $\piptwo$-complete.  }

\def\thmpiptwoextend{ Let $\scs = (\relset, \depset)$ be a schema such
	that $\relset$ consists of a single relation symbol $R$ and
	$\depset$ consists of two nontrivial FDs $X\rightarrow Y$ and
	$W\rightarrow Z$.  Suppose that each of $W$ and $Z$ contains an
	attribute that is in none of the other three sets. Then \gc is
	$\piptwo$-complete.}

\begin{theorem}\label{thm:piptwoextend}
	\thmpiptwoextend
\end{theorem}

As an example, recall that in $\scs^6$ we have
$\depset=\set{\emptyset\rightarrow A, B\rightarrow C}$.  This schema
is a special case of Theorem~\ref{thm:piptwoextend}, since we can use
$\emptyset\rightarrow A$ as $X\rightarrow Y$ and $B\rightarrow C$ as
$W\rightarrow Z$; and indeed, each of $W$ and $Z$ contains an
attribute (namely $B$ and $C$, respectively) that is not in any of the
other three sets. Additional examples of sets of FDs that satisfy the
conditions of Theorem~\ref{thm:piptwoextend} (and hence the
corresponding \gc is $\piptwo$-complete) follow. All of these sets are
over a relation symbol $R/4$. (And in each of these sets, the first FD
corresponds to $X\rightarrow Y$ and the second to $W\rightarrow Z$.)
\begin{compactitem}
	\item $A\rightarrow B$, $C\rightarrow D$
	\item $A\rightarrow C$, $AB\rightarrow CD$
	\item $A\rightarrow B$, $ABC\rightarrow D$
	\item $A\rightarrow B$, $C\rightarrow ABD$
\end{compactitem}

Unlike $\scs^6$, to this day we do not know what is the complexity of
g-categoricity$\angs{\scs^i}$ for any of the other $\scs^i$ (defined
in Section~\ref{sec:spec-schemas}). This includes $\scs^0$, for which
all we know is membership in coNP (as stated in
Theorem~\ref{cor:upperbounds}). However, except for this open problem,
the proof technique of Theorem~\ref{thm:pareto} is valid for
\gc. Consequently, we can show the following.

\def\thmgcatszerohypothetical{The following are equivalent.
	\begin{citemize}
		\item g-categoricity$\angs{\scs^0}$ is coNP-hard.
		\item g-categoricity$\angs{\scs}$ is coNP-hard for every schema $\scs$
		that falls outside the polynomial-time cases of
		Theorem~\ref{thm:global-1fd}.
	\end{citemize}}
	\begin{theorem}\label{thm:gcat-s0-hypothetical}
		\thmgcatszerohypothetical
	\end{theorem}

	\subsection{Transitive Priority}
	Let $(I,\graph,\succ)$ be an inconsistent prioritizing instance. We
	say that $\succ$ is \e{transitive} if for every two facts $f$ and $g$
	in $I$, if $f$ and $g$ are neighbors in $\graph$ and $f\succs g$, then
	$f\succ g$. Transitivity is a natural assumption when $\succ$ is
	interpreted as a partial order such as ``is of better quality than''
	or ``is more current than.'' In this section we consider
	g-categoricity in the presence of this assumption.  The following
	example shows that a g-repair is not necessarily a c-repair, even if
	$\succ$ is transitive. This example provides an important context for
	the results that follow.
	
	\begin{example}\label{example:x-repair-transitive}
		Consider again $I$ and $\succ$ from Example~\ref{example:ccat-fail}
		(depicted in Figure~\ref{fig:transitive-example}). Observe that
		$\succ$ is transitive. In particular, there is no priority between
		$\val{Ab2}$ and $\val{Ba2}$, even though $\val{Ab2}\succ\val{Ba2}$,
		because $\val{Ab2}$ and $\val{Ba2}$ are not in conflict (or put
		differently, they are not neighbors in
		$\graph_{\scs^1}^I$). Consider the following subinstance of $I$.
		\[ J\eqdef\set{\val{Aa1},\val{Ba2},\val{Ab3},\val{Bb1}}\]
		The reader can verify that $J$ is a g-repair, but not a c-repair
		(since no execution of $\algname{FindCRep}$ can generate $J$).
	\end{example}

	Example~\ref{example:x-repair-transitive} shows that the notion global
	optimality is different from completion optimality, even if the
	priority relation is transitive. Yet, quite remarkably, the two
	notions behave the same when it comes to categoricity.
	
	\begin{theorem}\label{thm:g-c-same-transitive}
		Let $D=(I,\graph,\succ)$ be an inconsistent prioritizing instance
		such that $\succ$ is transitive. $|\crep{D}|=1$ if and only if
		$|\grep{D}|=1$.
	\end{theorem}
	\begin{proof}
		The ``if'' direction follows from
		Proposition~\ref{prop:containments}, since every c-repair is also a
		g-repair.  The proof of the ``only if'' direction is based on the
		special structure of the c-repair, as established in
		Section~\ref{sec:c}, in the case where only one c-repair
		exists. Specifically, suppose that there is a single c-repair $J$
		and let $J'\neq J$ be a consistent subinstance of $I$.  We need to
		show that $J'$ has a global improvement. We claim that $J$ is a
		global improvement of $J'$. This is clearly the case if $J'\subseteq
		J$. So suppose that $J'\not\subseteq J$.  Let $f'$ be a fact in
		$J'\setminus J$. We need to show that there is a fact $f\in
		J\setminus J'$ such that $f\succ f'$. We complete the proof by
		finding such an $f$.
		
		Recall from Theorem~\ref{thm:ccategoricity-correct} that $J$ is the
		result of executing
		$\algname{CCategoricity}(I,\graph,\succ)$. Consider the positive
		strata $P_i$ and the negative strata $N_j$ constructed in that
		execution. Since $J$ is the union of the positive strata, we get
		that $f'$ necessarily belongs to a negative stratum, say $N_j$. From
		the definition of $N_j$ it follows that $\graph$ has a hyperedge $e$
		such that $f'\in e$,
		$(e\setminus\set{f'})\subseteq P_1\cup\dots\cup P_j$, and
		$(e \setminus \set{f'})\succs f'$. Let $e$ be such a
		hyperedge. Since $J'$ is consistent, it cannot be the case that $J'$
		contains all the facts in $e$. Choose a fact $f\in e$ such that
		$f\notin J'$. Then $f\succs f'$, and since $\succ$ is transitive
		(and $f$ and $f'$ are neighbors), we have $f\succ f'$.  So
		$f\in J\setminus J'$ and $f\succ f'$, as required.
	\end{proof}

	Combining Theorems~\ref{thm:c-categoricity-ptime}
	and~\ref{thm:g-c-same-transitive}, we get the following.
	
	\begin{corollary}\label{cor:g-categoricity-transitive-ptime}
		For transitive priority relations, problems g-categoricity and
		c-categoricity coincide, and in particular,
		g-categoricity is solvable in polynomial time.
	\end{corollary}
	
	\begin{commentqed}
		The reader may wonder whether Theorem~\ref{thm:g-c-same-transitive}
		and Corollary~\ref{cor:g-categoricity-transitive-ptime} hold for
		p-categoricity as well. This is not the case. The hardness of
		p-categoricity$\angs{\scs^6}$
		(Theorem~\ref{thm:hardness-specific-pareto}) is proved by
		constructing a reduction where the priority relation is transitive
		(and in fact, it has no chains of length larger than one).
	\end{commentqed}
	
	In their analysis, Fagin et al.~\cite{DBLP:conf/pods/FaginKK15} have
	constructed various reductions for proving coNP-hardness of g-repair
	checking. In several of these, the priority relation is transitive. We
	conclude that there are schemas $\scs$ such that, on transitive
	priority relations, g-repair checking is coNP-complete whereas
	g-categoricity is solvable in polynomial time.

\section{Related Work on Data Cleaning}~\label{sec:related} 

We now discuss the relationship between our work and past work on data
cleaning. Specifically, we focus on relating and contrasting our
complexity results with ones established in past research.  To the
best of our knowledge, there has not been any work on the complexity
of categoricity within the prioritized repairing of Staworko et
al.~\cite{DBLP:conf/edbtw/StaworkoCM06}.  Fagin et
al.~\cite{DBLP:conf/pods/FaginKRV14} investigated a static version of
categoricity in the context of text extraction, but the settings and
problems are very different, and so are the complexity results (e.g.,
Fagin et al.~\cite{DBLP:conf/pods/FaginKRV14} establish undecidability
results).

Bohannon et al.~\cite{DBLP:conf/sigmod/BohannonFFR05} have studied a
repairing framework where repairing operations involve attribute
updates and tuple insertions, and where the quality of a repair is
determined by a cost function (aggregating the costs of individual
operations). They have shown that finding an optimal repair is
NP-hard, in data complexity, even when integrity constraints consist
of only FDs. This result could be generalized to hardness of
categoricity in their model (e.g., by a reduction from the \e{unique
	exact 3-cover} problem~\cite{DBLP:journals/iandc/SchulteLG10}).  The
source of hardness in their model is the cost minimization, and it is
not clear how any of our hardness results could derive from those, as
the framework of preferred repairs (adopted here) does not involve any
cost-based quality; in particular, as echoed in this paper, an optimal
repair can be found in polynomial time under each of the three
semantics~\cite{DBLP:journals/amai/StaworkoCM12}.

In the framework of \e{data currency}~\cite[Chapter
6]{DBLP:series/synthesis/2012Fan}\cite{DBLP:journals/tods/FanGW12},
relations consist of entities with attributes, where each entity may
appear in different tuples, every time with possibly different
(conflicting) attribute values. A partial order on each attribute is
provided, where ``greater than'' stands for ``more current.'' A
\e{completion} of an instance is obtained by completing the partial
order on an attribute of every entity, and it defines a \e{current
	instance} where each attribute takes its most recent value. In
addition, a completion needs to satisfy given (denial) constraints,
which may introduce interdependencies among completions of different
attributes. Fan et al.~\cite{DBLP:journals/tods/FanGW12} have studied
the problem of determining whether such a specification induces a
single current instance (i.e., the corresponding version of
categoricity), and showed that this problem is coNP-complete under
data complexity.  It is again not clear how to simulate their hardness
in our p-categoricity and g-categoricity, since their hardness is due
to the constrains on completions, and these constraints do not have
correspondents in our case (beyond the partial orders). A similar
argument relates our lower bounds to those in the framework of
conflict resolution by Fan
Geerts~\cite[Chapter~7.3]{DBLP:series/synthesis/2012Fan}, where the
focus is on establishing a unique tuple from a collection of
conflicting tuples.

Fan et al.~\cite{DBLP:journals/tods/FanGW12} show that in the absence
of constraints, their categoricity problem can be solved in polynomial
time (even in the presence of ``copy functions'').  This tractability
result can be used for establishing the tractability side of
Theorem~\ref{thm:pareto} in the special case where the single FD is a
key constraint. In the general case of a single FD, we need to argue
about relationships among sets, and in particular, the differences
among the three x-categoricity problems matter.

Cao et al.~\cite{DBLP:conf/sigmod/CaoFY13} have studied the problem of
entity record cleaning, where again the attributes of an entity are
represented as a relation (with missing values), and a partial order
is defined on each attribute. The goal is to increase the accuracy of
values from the partial orders and an external source of reliable
(``master'') data.  The specification now gives update steps that have
the form of logical rules that specify when one value should replace a
null, when new preferences are to be derived, and when data should be
copied from the master data. Hence, cleaning is established by
\e{chasing} these rules. They study a problem related to categoricity,
namely the \e{Church-Rosser} property: is it the case that every
application of the chase (in any rule-selection and grounding order)
results in the same instance?  They show that this property is
testable in polynomial time by giving an algorithm that tests whether
some invalid step in the end of the execution has been valid sometime
during the execution. We do not see any clear way of deriving any of
our upper bounds from this result, due to the difference in the update
model (updating nulls and preferences vs.~tuple deletion), and the
optimality model (chase termination vs.~x-repair).

The works on \e{certain
	fixes}~\cite{DBLP:journals/jdiq/FanM0Y14,DBLP:journals/vldb/FanLMTY12}
\cite[Chapters~7.1--7.2]{DBLP:series/synthesis/2012Fan} consider
models that are substantially different from the one adopted here,
where repairs are obtained by chasing update rules (rather than tuple
deletion), and uniqueness applies to chase outcomes (rather than
maximal subinstances w.r.t.~preference lifting). The problems relevant
to our categoricity are the \e{consistency
	problem}~\cite{DBLP:journals/jdiq/FanM0Y14} (w.r.t.~guarantees on
the consistency of some attributes following certain patterns), and
the \e{determinism} problem~\cite{DBLP:journals/jdiq/FanM0Y14}. They
are shown to be intractable
(coNP-complete~\cite{DBLP:journals/vldb/FanLMTY12} and
PSPACE-complete~\cite{DBLP:journals/jdiq/FanM0Y14}) under combined
complexity (while we focus here on data complexity).

Finally, we remark that there have several dichotomy results on the
complexity of problems associated with inconsistent
data~\cite{DBLP:journals/jcss/MaslowskiW13,DBLP:conf/pods/FaginKK15,DBLP:conf/icdt/KoutrisS14},
but to the best of our knowledge this paper is the first to establish
a dichotomy result for any variant of repair uniqueness identification.

\section{Concluding Remarks}\label{sec:conclusions}
We investigated the complexity of the categoricity problem, which is
that of determining whether the provided priority relation suffices to
clean the database unambiguously in the framework of preferred
repairs. Following the three semantics of optimal repairs, we
investigated the three variants of this problem: p-categoricity,
g-categoricity and c-categoricity. We established a dichotomy in the
data complexity of p-categoricity for the case where constraints are
FDs, partitioning the cases into polynomial time and
coNP-completeness. We further showed that the tractable side of
p-categoricity extends to g-categoricity, but the latter can reach
$\piptwo$-completeness already for two FDs. Finally, we showed that
c-categoricity is solvable in polynomial time in the general case
where integrity constraints are given as a conflict hypergraph.  We
complete this paper by discussing directions for future research.

In this work we did not address any qualitative discrimination among
the three notions of x-repairs. Rather, we continue the line of
work~\cite{DBLP:conf/edbtw/StaworkoCM06,DBLP:conf/pods/FaginKRV14}
that explores the impact of the choice on the entailed computational
complexity. It has been established that, as far as repair checking is
concerned, the Pareto and the completion semantics behave much better
than the global one, since g-repair checking is tractable only in a
very restricted class of schemas~\cite{DBLP:conf/pods/FaginKRV14}. In
this work we have shown that from the viewpoint of categoricity, the
Pareto semantics departs from the completion one by being likewise
intractable (while the global semantics hits an even higher complexity
class), hence the completion semantics outstands so far as the most
efficient option to adopt.

It would be interesting to further understand the complexity of
g-categoricity, towards a dichotomy (at least for FDs).  We have left
open the question of whether there exists a schema with a single
relation and a set of FDs, \e{not} equivalent to a single FD, such
that g-categoricity is solvable in polynomial time.  Beyond that, for
both p-categoricity and g-categoricity it is important to detect
islands of tractability based on properties of the data and/or the
priority relation (as schema constraints do not get us far in terms of
efficient algorithms, at least by our dichotomy for p-repairs), beyond
transitivity in the case of g-categoricity
(Corollary~\ref{cor:g-categoricity-transitive-ptime}).

Another interesting direction would be the generalization of
categoricity to the problems of \e{counting} and \e{enumerating} the
preferred repairs. For classical repairs (without a priority
relation), Maslowski and
Wijsen~\cite{DBLP:journals/jcss/MaslowskiW13,DBLP:conf/icdt/MaslowskiW14}
established dichotomies (FP vs.~\#P-completeness) in the complexity of
counting in the case where constraints are primary keys. For the
general case of denial constraints, counting the classical repairs
reduces to the enumeration of independent sets of a hypergraph with a
bounded edge size, a problem shown by Boros et
al.~\cite{DBLP:journals/ppl/BorosEGK00} to be solvable in incremental
polynomial time (and in particular polynomial input-output
complexity). For a general given conflict hypergraph, repair
enumeration is the well known problem of enumerating the \e{minimal
	hypergraph transversals} (also known as the hypergraph \e{duality}
problem); whether this problem is solvable in polynomial total time is
a long standing open problem~\cite{DBLP:conf/csl/GottlobM14}.

In this work we focused on cleaning within the framework of preferred
repairs, where integrity constraints are anti-monotonic and cleaning
operations are tuple deletions (i.e., \e{subset repairs}). However,
the problem of categoricity arises in every cleaning framework that is
based on defining a set of repairs with a preference between repairs,
including different types of integrity constraints, different cleaning
operations (e.g., tuple addition and cell
update~\cite{DBLP:journals/tods/Wijsen05}), and different priority
specifications among repairs.  This includes preferences by means of
general scoring
functions~\cite{DBLP:conf/iqis/MotroAA04,DBLP:books/idea/encyclopediaDB2005/GrecoSTZ05},
aggregation of scores on the individual cleaning
operations~\cite{DBLP:conf/sigmod/BohannonFFR05,DBLP:journals/jdiq/FanM0Y14,DBLP:conf/icdt/KolahiL09,DBLP:journals/jdiq/FanM0Y14,DBLP:conf/sigmod/DallachiesaEEEIOT13},
priorities among resolution
policies~\cite{DBLP:conf/kr/MartinezPPSS08} and preferences based on
soft
rules~\cite{DBLP:journals/tplp/NieuwenborghV06,DBLP:conf/dexa/GrecoSTZ04}.
This also includes the LLUNATIC
system~\cite{DBLP:journals/pvldb/GeertsMPS13,DBLP:conf/icde/GeertsMPS14}
where priorities are defined by lifting partial orders among ``cell
groups,'' representing either semantic preferences (e.g., timestamps)
or level of completeness (e.g., null vs.~non-null).  A valuable future
direction would be to investigate the complexity of categoricity in
the above frameworks, and in particular, to see whether ideas or proof
techniques from this work can be used to analyze their categoricity.

\eat{A complexity result of PSPACE-completeness
	has been established by Fan et al.~\cite{DBLP:journals/jdiq/FanM0Y14}
	(where the problem is referred to as ``determinism'') and it would be
	interesting to pursue a refinement of that result.}

Motivated by the tractability of c-categoricity, we plan to pursue an
implementation of an interactive and declarative system for database
cleaning, where rules are of two kinds: integrity constraints and
priority specifications (e.g., based on the semantics of \e{priority
	generating dependencies} of Fagin et al.~\cite{DBLP:conf/pods/FaginKRV14}). To make
such a system applicable to a wide range of practical use cases, we will
need to extend beyond subset repairs, and consequently, investigate
the fundamental direction of extending the framework of preferred
repairs towards such repairs.

\newpage
{
\balance
\bibliographystyle{abbrv} 
\bibliography{main}  
}

\onecolumn
\appendix
\def\refpareto{\ref{sec:p}}

\newcommand{\p}{^\prime}
\newcommand{\Ia}{I_{\tup a}}
\newcommand{\Iab}{I_{\tup {a,b}}}
\newcommand{\Jab}{K_{\tup {a,b}}}
\newcommand{\Iabp}{I_{\tup {a,b\p}}}
\newcommand{\Iac}{I_{\tup {a,c}}}

\newcommand{\pcp}{p-categoricity problem\xspace}
\newcommand{\pr}{p-repairs\xspace}

\newcommand{\allXxu}{for all $X\in \mathcal{X}$, $x\in X$ and $u\in \mathcal{U}$\xspace }

\section{Proofs for Section~\refpareto}

In this section we provide proofs for Section~\ref{sec:p}.

\subsection{Proof of Lemma~\ref {lemma:key-for-pareto-1fd}}
In the following section, we say that a block $\Ia$ (respectively,
$\Iab$) is the block (respectively, subblock) of a fact $f$ if
$f\in \Ia$ (respectively $f\in \Iab$).  Note that each fact has a
unique block and subblock.
\begin{replemma}{\ref{lemma:key-for-pareto-1fd}}
	\lemmakeyforparetoonefd
\end{replemma}
\begin{proof}
	Recall that $\Delta$ is the set $\set{A\rightarrow B}$.  We start by
	proving the second part of the lemma.  That is, we show that each
	p-repair of a block $\Ia$ is a subblock $\Iab$.  Let $K$ be a
	p-repair of $\Ia$.  Then $K$ is contained in a single subblock of
	$\Ia$, since $K$ is consistent. Moreover, $K$ contains all the facts
	in $\Iab$, or otherwise $K$ has a Pareto improvement.
	
	Next, we prove the first part of the lemma.
	\partitle{The ``if'' direction}
	Let $J$ be a p-repair of $I$. We need to show that $J$ is a union of
	p-repairs over all the blocks $\Ia$ of $I$.  Observe that $J$ is
	consistent, and so, for each block $\Ia$ it contains facts from at
	most one subblock $\Iab$. Moreover, since $J$ is maximal, it contains
	at least one representative from each block $\Ia$, and furthermore, it
	contains the entire subblock of each such a representative. We
	conclude that $J$ is the union of subblocks of $I$. It is left to show
	that if a subblock $\Iab$ is contained in $J$, then $\Iab$ is a
	p-repair of $\Ia$. Let $\Iab$ be a subblock contained in $J$ and
	assume, by way of contradiction, that $K$ is a Pareto-improvement of
	$\Iab$ in $\Ia$. Let $J'$ be the instance that is obtained from $J$ by
	replacing $\Iab$ with $K$. Observe that $J'$ is consistent, since no
	facts in $J$ other than those in $\Iab$ conflict with facts from $K$.
	Then clearly, $J'$ is a Pareto improvement of $J$, which contradicts
	the fact that $J$ is a p-repair.
	
	\partitle{The ``only if'' direction}
	Let $J$ be a union of p-repairs over all the blocks $\Ia$ of $I$.  We
	need to show that $J$ is a p-repair.  By the second part of the lemma,
	$J$ is a union of subblocks.  Since each subblock is consistent and
	facts from different blocks are consistent, we get that $J$ is
	consistent.  It is left to show that $J$ does not have a Pareto
	improvement.  Assume, by way of contradiction, that $J$ has a Pareto
	improvement $K$. By the definition of a Pareto improvement, $K$
	contains a fact $f$ such that $f\succ g$ for all $g\in J\setminus K$.
	Let $f$ be such a fact. Let $\Ia$ be the subblock of $f$. Then, by our
	assumption the subinstance $J$ contains a p-repair of $\Ia$, and from
	the second part of the lemma this p-repair is a subblock of $\Ia$, say
	$\Iab$.  But then, $f$ is not in $\Iab$ (since $f\notin J$), and
	therefore, $K$ does not contain any fact from $\Iab$ (since $K$ is
	consistent). We conclude that $f\succ g$ for all $g\in \Iab$, and
	hence, $\Iab$ has a Pareto improvement (namely $\set{f}$), in
	contradiction to the fact that $\Iab$ is a p-repair of $\Ia$.
\end{proof}

\subsection{Proof of Theorem~\ref{thm:hardness-specific-pareto}}

We now prove Theorem~\ref{thm:hardness-specific-pareto}

\begin{reptheorem}{\ref{thm:hardness-specific-pareto}}
	\theoremhardnessspecificpareto
\end{reptheorem}

We give a separate proof for each of the two schemas.
\newcommand{\XC}{\mathsf{XC}}
\subsubsection{Hardness of p-categoricity$\angs{\scs^0}$}
We construct a reduction from the Exact-Cover problem ($\XC$) to the complement of p-categoricity$\angs{\scs^0}$.
The input to $\XC$ is a set $\mathcal{U}$  of elements and a collection 
$\mathcal{X}$ of subsets of $\mathcal{U}$, such that their union is $\mathcal{U}$.
The goal is to identify whether there is an exact cover of $\mathcal{U}$ by $\mathcal{X}$.
An \e{exact cover} of $\mathcal{U}$ by $\mathcal{X}$ is a collection of pairwise disjoint sets from $\mathcal{X}$ whose union is $\mathcal{U}$.

\partitle{Construction}
Given an input $(\mathcal{X},\mathcal{U})$ to $\XC$, we construct input
$(I,\succ)$ for p-categoricity$\angs{\scs^0}$.
For each $u \in \mathcal{U}$, 
$X\in \mathcal{X}$ and 
$x\in X$, the instance
$I$ consists of the following facts:
\begin{enumerate}
	\item[$(i)$]
	${R^0}(u,u)$
	\item[$(ii)$]
	${R^0}(u,f_u)$
	\item[$(iii)$]
	${R^0}(X_x,f_u)$ 
	\item[$(iv)$]
	${R^0}(X_x,X_x)$
	\item[$(v)$]
	${R^0}(X_x, x)$
	\item[$(vi)$] ${R^0}(X_{x_{i+1}}, x_i )$ for each $i=0,\ldots,n-1$,
	where $X=\{x_0,\ldots,x_{n-1}\}$ and plus is interpreted modulo $n$
	(e.g., $(n-1)+1=0$)
\end{enumerate}
In the sequel, we relate to these facts by \e{types} according to their roman number. For example, facts of the form $R^0(X_x,f_u)$, where $X\in \mathcal{X}$, $x\in X$ and $u\in \mathcal{U}$, will be referred to as facts of type $(iii)$.

For all $u\in \mathcal{U}, X\in \mathcal{X}$ and $x\in X$,
the priority relation $\succ$ is defined as follows:
\begin{itemize}
	\item
	${R^0}(u,u) \succ {R^0}(u,f_u)$,
	\item
	${R^0}(X_x,X_x) \succ {R^0}(X_x,f_u)$,
	\item
	${R^0}(u,u) \succ {R^0}(X_x,u)$,
	\item
	${R^0}(X_x,f_u) \succ {R^0}(X_x,x)$, and
	\item
	${R^0}(X_{x_{i+1}}, x_i ) \succ {R^0}(X_{x_i},x_i)$, 
	for $i=0,\ldots,n-1$ where $X=\{x_0,\ldots,x_{n-1}\}$
	and plus is interpreted modulo $n$.
	Note that each $X\in \mathcal{X} $ has a corresponding $n$.
\end{itemize}

	\begin{figure}[t]
		\centering
		\input{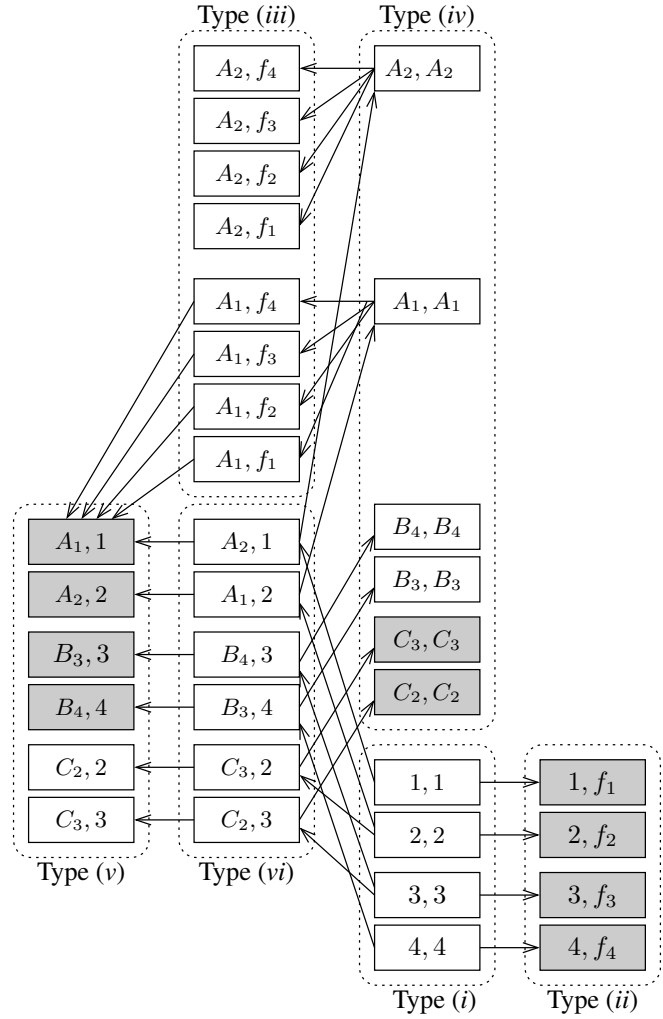}
		\caption{\label{fig:paretoHardness}Illustration of the reduction from $\XC$ to p-categoricity$\angs{\scs^0}$}
	\end{figure}

Our construction is partly illustrated in Figure~\ref{fig:paretoHardness} for the following input to the $\XC$ problem:
$\mathcal{U}= \{1,2,3,4\}$, $\mathcal{X}=\{A,B,C\}$ where
$A=\{1,2\}$, $B=\{3,4\}$ and $C=\{2,3\}$.
Note that we  denote the fact $R^0(a,b)$ by  $(a,b)$.
In this case, there is an exact cover of $\mathcal{U}$ by $\mathcal{X}$ that consists of the sets $A$ and $B$.
The gray facts represent a p-repair. 
\partitle{Proof of Hardness}
We start by finding a c-repair.
\begin{lemma}
	\label{pareto:l:greedy}
	There is a c-repair that consists of the following facts for all $u\in \mathcal{U}$, $X\in \mathcal{X}$ and $ x\in X$.
	\begin{itemize}
		\item
		${R^0}(u,u)$
		\item
		${R^0}(X_x,X_x)$
	\end{itemize}
\end{lemma}
\begin{proof}
	It is straightforward to show that
	there is a run of the algorithm
	$\algname{FindCRep}(I,\graph,\succ)$ that returns exactly  this c-repair.
\end{proof}
\newcommand{\Jz}{J_0}
In the the remainder of the proof, we relate to the the c-repair from Lemma~\ref{pareto:l:greedy} by $J_0$. Note that every c-repair is also a p-repair and therefore $\Jz$ is a p-repair of $I$.
To complete the proof we show that there is a solution to $\XC$ if and only if $I$ has a p-repair different from $\Jz$.
\newcommand{\tl}[1]{\tilde{#1}}
\newcommand{\xx}{X_x}
\newcommand{\fu}{f_u}
\newcommand{\xxr}[2]{{#1}_{#2}}
\partitle{The ``if'' direction}
We show that if there is a solution to $\XC$ then $I$ has a p-repair
different from $\Jz$.

We construct a p-repair of $I$, namely $K$, different from $\Jz$ based on a solution to $\XC$.
Let the collection of sets
$X^{1},\ldots X^{l}\in \mathcal{X}$ be a solution to $\XC$.
Let $K$ consist of the following facts for all 
$X\in \mathcal{X}$, $x\in X$ and $u\in \mathcal{U}$.
\begin{itemize}
	\item[$(1)$] 
	${R^0}(\xx,x)$ if $X\in\{X^1,\ldots,X^l\}$ 
	\item[$(2)$] 
	${R^0}(X_x,X_x)$ if $X\not \in\{X^1,\ldots,X^l\}$
	\item[$(3)$]
	${R^0}(u,f_u)$ 
\end{itemize}
Note that since $K$ is different from $\Jz$ (see Lemma~\ref{pareto:l:greedy}), it is left to show that $K$ is a p-repair of $I$.
To do so, we show that $K$ is consistent and that it does not have a Pareto improvement.

\begin{lemma}
	\label{pareto:l:K consistent}
	$K$ is a consistent subinstance of $I$.
\end{lemma}
The proof of this lemma is straightforward based on the above construction.
\eat{\begin{proof}
		In order for two facts to be inconsistent, they  should agree (disagree) on $A$ but disagree (agree) on $B$.
		Obviously, a fact of type $R^0(u,f_u)$ for some $u\in \mathcal{U}$, is consistent with the fact $R^0(X_x,x)$ and $R^0(X_x,X_x)$ for any $X\in\mathcal{X},x\in X$.
		Moreover the facts $R^0(X_x,x)$ where $X\in\{X^1,\ldots,X^l\},x\in X$ and $R^0(Y_y,Y_y)$ where $X\in\mathcal{X} \setminus \{X^1,\ldots,X^l\},x\in X$ are consistent.
		Hence, each two facts in $K$ are consistent. That implies that $K$ is consistent.
	\end{proof}}
	\begin{lemma}
		\label{pareto:l:K optimal}
		$K$ does not have a Pareto improvement.
	\end{lemma}
	\def\ititle#1{\textbf{#1:}\,\,}
	
	\begin{proof}
		It suffices to show that for all $f$ in $I\setminus K$, there exists $f\p$ in $K$ such that $\{f,f\p\}$ is inconsistent (w.r.t. $\Delta^0$) and $f\not \succ f\p$.
		For each $f\in I \setminus K$ we choose $f\p \in K$ such that the conditions hold.
		We divide to different cases according to the type of $f$.
		\begin{itemize}
			\item \ititle{$f$ is of type $(i)$} That is, there exists an element
			$u$ in $\mathcal{U}$ such that $f={R^0}(u,u)$.  Since the collection
			$\{X^1,\ldots, X^l\}$ is a cover of $\mathcal{U}$, there exists $X$
			in $\{X^1,\ldots, X^l\}$ such that $u\in X$. Hence,
			$R^0(X_u,u)\in K$ and we choose $f\p = R^0(X_u,u)$.
			\item \ititle{$f$ is of type $(ii)$} This is impossible, since $K$
			contains all the facts of this type.
			\item
			\ititle{$f$ is of type $(iii)$} 
			That is, there exists $X\in \mathcal{X},$ $x\in X $ and $u\in \mathcal{U}$ such that $f={R^0}(X_x,f_u)$. 
			We choose $f\p ={R^0}(u,f_u)$. 
			\item
			\ititle{$f$ is of type $(iv)$} 
			That is, there exists a set $X \in \mathcal{X}$ and $x\in X$ such that $f=R^0(X_x,X_x)$.
			Since $f\not \in K$, it holds that $X\in \{X^1,\ldots,X^l\}$.
			Hence $R^0(X_x,x)\in K$ and
			we choose $f\p = R^0(X_x,x)$.
			\item
			\ititle{$f$ is of type $(v)$} 
			That is, there exists a set $X \in \mathcal{X}$ and $x\in X$ such that $f=R^0(X_x,x)$.
			Since $f\not \in K$, it holds that $X\not \in \{X^1,\ldots,X^l\}$. Hence $ R^0(X_x,X_x)\in K$ and
			we choose $f\p = R^0(X_x,X_x)$.
			\item \ititle{$f$ is of type $(vi)$} That is, there exists a set
			$X \in \mathcal{X}$ where $X=\{x_0,\ldots,x_{n-1}\}$ and
			$0\le i \le n-1$ such that $f=R^0(X_{x_{i+1}},x_i)$.  Since the
			collection $\{X^1,\ldots, X^l\}$ is a cover of $\mathcal{U}$, there
			exists some $Y$ in $\{X^1,\ldots, X^l\}$ such that $x_i\in Y$.
			Hence, we get that $R^0(Y_{x_i},x_i)\in K$ and we choose
			$f\p = R^0(Y_{x_i},x_i)$.
		\end{itemize}
		It holds that in all of the above cases $\{f,f\p\}$ is inconsistent and $f \not \succ f\p$.
	\end{proof}
	
	Lemmas~\ref{pareto:l:K consistent} and~\ref{pareto:l:K optimal} imply
	the following.
	\begin{lemma}
		$K$ is a p-repair of $I$.
	\end{lemma}
	Given a solution to $\XC$, we constructed a p-repair $K$ different from the above c-repair. This completes the ``if'' direction.

	\begin{figure}[t]
		\centering
		\input{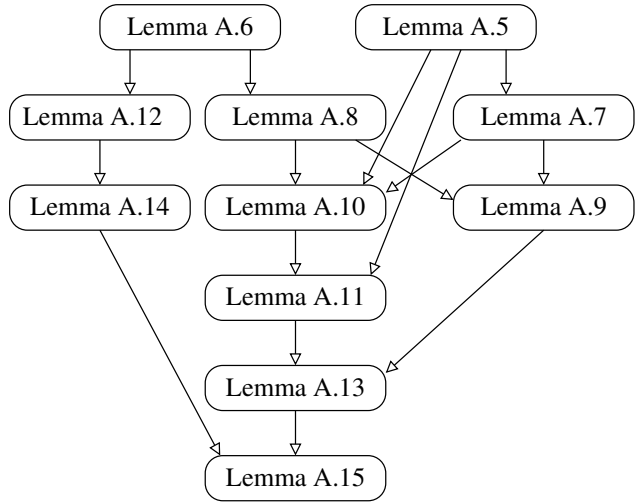}
		\caption{\label{fig:lemmas-onlyif}Dependencies between the lemmas in the proof of the ``only if'' direction of p-categoricity$\angs{\scs^0}$}
	\end{figure}

	\partitle{The ``only if'' direction}
	We show that if $I$ has a p-repair different from $\Jz$ (from
	Lemma~\ref{pareto:l:greedy}) then there is a solution to $\XC$.  The
	proof of this direction consists of several lemmas, and the
	dependencies between them are described in
	Figure~\ref{fig:lemmas-onlyif}. For example, the proof of
	Lemma~\ref{pareto:l: all (u,u) together} is based on
	Lemmas~\ref{pareto:l:notmiddle R2(Xx,f_u)},~\ref{pareto:l: (u,u) not
		in K, (u,fu) in K} and~\ref{pareto:l: (u,f_u) in K, (Xu,u) in K}.
	
	Let $K$ be a p-repair different from $\Jz$.  In the next two lemmas we
	show that facts of types $(iii)$ and $(vi)$ are not in $K$.
	\begin{lemma}
		\label{pareto:l:notmiddle R2(Xx,f_u)}
		For all $X\in \mathcal{X}$, $x\in X$ and $u\in \mathcal{U}$, it holds that
		${R^0}(\xx,f_u) \not \in K$.
	\end{lemma}
	\begin{proof}
		Let 
		${X}\in \mathcal{X}$, ${x}\in X$ and ${u}\in \mathcal{U}$.
		Assume, by way of contradiction, that 
		${R^0}(\xx,\fu) \in K$.
		Since $K$ is consistent and ${R^0}(\xx,\xx)$
		is inconsistent with ${R^0}(\xx,\fu)$, we obtain that
		${R^0}(\xx,\xx)\not \in K$.  
		Since ${R^0}(\xx,\xx) \succ {R^0}(\xx,\fu)$, it holds that
		$K$ must contain a fact $f$ that is inconsistent with ${R^0}(\xx,\xx)$ (but is consistent with ${R^0}(\xx,\fu)$). 
		Therefore, $f$ must agree with ${R^0}(\xx,\xx)$ 
		on $B$. 
		That leads to a contradiction since there is no such fact.
	\end{proof}
	
	\begin{lemma}
		\label{pareto:l:notmiddle R2(Xx,u)}
		For all $X\in \mathcal{X}$, $x\in X$ and $u\in \mathcal{U}$ where $u\ne x$, it holds that
		${R^0}(\xx,u)\not \in K$.
	\end{lemma}
	\begin{proof}
		Let 
		${X}\in \mathcal{X}$, ${x}\in X$ and ${u}\in \mathcal{U}$ where $u\ne x$.
		Assume that ${R^0}(\xx,u) \in K$. 
		Since \allXxu, we have that ${R^0}(u,u) \succ {R^0}(\xx,u)$, the p-repair $K$ must contain a fact $f$ that is inconsistent with ${R^0}(u,u)$ (but is consistent with ${R^0}(\xx,u)$). 
		Since $K$ is consistent, $f$ must agree with ${R^0}(u,u)$ on $A$. 
		Thus, $f$ must be the fact ${R^0}(u,f_u)$.
		Replacing both facts ${R^0}(\xx,u)$ and ${R^0}(u,f_u)$ with ${R^0}(u,u)$ results in a Pareto improvement of $K$ which leads to a contradiction since $K$ is a p-repair. 
	\end{proof}

	We establish a connection between facts of type $(i)$ and $(ii)$.
	\begin{lemma}
		\label{pareto:l: (u,u) not in K, (u,fu) in K}
		Let $u\in \mathcal{U}$. 
		If the fact ${R^0}(u,u)\not \in K$ then  ${R^0}(u,f_u)\in K$.
	\end{lemma}
	\begin{proof}
		Assume ${R^0}(u,u)\not \in K$. 
		Assume, by way of contradiction, that ${R^0}(u,f_u)\not \in K$.
		Since $K$ is maximal, it must contain a fact $f$ that is inconsistent with ${R^0}(u,f_u)$. 
		If $f$ agrees with ${R^0}(u,f_u)$ on $A$ then it can only be of type $(i)$. That is, the only possibility is that $f={R^0}(u,u)$ which leads to a contradiction.
		If $f$ agrees with ${R^0}(u,f_u)$ on $B$, it can only be of type $(iii)$ which leads to a contradiction since by Lemma~\ref{pareto:l:notmiddle R2(Xx,f_u)}, such a fact cannot be in a p-repair.
	\end{proof}
	Moreover, we establish a connection  between facts of type $(ii)$ and $(v)$.
	\begin{lemma}
		\label{pareto:l: (u,f_u) in K, (Xu,u) in K}
		Let $u\in \mathcal{U}$. If
		${R^0}(u,f_u)\in K$
		then there exists $X\in \mathcal{X}$ such that $u\in X$ and ${R^0}(X_u,u)\in K$.
	\end{lemma}
	\begin{proof}
		Assume ${R^0}(u,f_u)\in K$.
		Since $K$ is consistent, ${R^0}(u,u)\not \in K$.
		It follows from
		${R^0}(u,u)\succ {R^0}(u,f_u)$ that
		$K$ must contain a fact $f$ that is inconsistent with  ${R^0}(u,u)$.
		Since $K$ is consistent, $f$ is consistent with ${R^0}(u,f_u)$.
		Therefore, $f$ must agree with ${R^0}(u,u)$ on $B$.
		The possible types for $f$ are $(vi)$ and $(v)$.
		By Lemme~\ref{pareto:l:notmiddle R2(Xx,u)}, $f$ is not of type $(vi)$.
		Thus, $f$ must be of type $(v)$.
		Therefore, ${R^0}(X_u,u)\in K$ for  $X\in \mathcal{X}$ such that $u\in X$ (there exists such $X$ since the union of sets of $\mathcal{X}$ is $\mathcal{U}$).
	\end{proof}
	We conclude the following connection  between facts of types $(i)$ and $(v)$.
	\begin{lemma}
		\label{pareto:l: (X_x,x) not in K, (x,x) in K}
		Let $u\in \mathcal{U}$. If
		${R^0}(X_u,u) \not \in K$ for all $X \in \mathcal{X}$ such that $u\in X$, then ${R^0}(u,u)\in K$.
	\end{lemma}
	\begin{proof}
		Assume ${R^0}(X_u,u) \not \in K$ for all $X \in \mathcal{X}$.
		Assume, by way of contradiction, that ${R^0}(u,u)\not \in K$.  By
		Lemma~\ref{pareto:l: (u,u) not in K, (u,fu) in K},
		${R^0}(u,f_u) \in K$.  By Lemma~\ref{pareto:l: (u,f_u) in K, (Xu,u)
			in K}, there exists $X\in \mathcal{X}$ such that $u\in X$ and
		${R^0}(X_u,u) \in K$.  This is a contradiction.
	\end{proof}
	We state that either $R^0(u,u)\in K$ for all $u\in \mathcal{U}$ or
	none of the facts $R^0(u,u)$, where $u\in \mathcal{U}$, is in $K$.
	\begin{lemma}
		\label{pareto:l: all (u,u) together}
		Let $u_0 \in \mathcal{U}$.
		If ${R^0}(u_0,u_0) \in K$ then 
		${R^0}(u,u) \in K$ for all $u \in \mathcal{U}$.
	\end{lemma}
	\begin{proof}
		Let $u_0 \in \mathcal{U}$ and assume ${R^0}(u_0,u_0)\in K$.  Assume,
		by way of contradiction, there exists $u\in \mathcal{U}$ such that
		$u \ne u_0$ and ${R^0}(u,u) \not \in K$.  By Lemma~\ref{pareto:l:
			(u,u) not in K, (u,fu) in K}, ${R^0}(u,f_u) \in K$.  Thus, by
		Lemma~\ref{pareto:l: (u,f_u) in K, (Xu,u) in K} we have that
		${R^0}(U_u,u) \in K$ for some $U\in \mathcal{X}$ such that $u\in U$.
		Note that since ${R^0}(U_u,f_{u_0}) \succ {R^0}(U_u,u)$, there must
		be a fact in $K$ that is inconsistent with ${R^0}(U_u,f_{u_0})$ and
		is consistent with ${R^0}(U_u,u)$.  The only such a fact is
		${R^0}(u_0,f_{u_0})$ (i.e, ${R^0}(u_0,f_{u_0})\in K$). This is a
		contradiction since ${R^0}(u_0,u_0)\in K$ and $K$ is consistent.
	\end{proof}
	
	We prove that all facts of type $(i)$ are not in $K$. 
	\begin{lemma}
		\label{pareto:l: (u,u) not in K}
		For all $u\in \mathcal{U}$, 
		${R^0}(u,u) \not \in K$.
	\end{lemma}
	\begin{proof}
		Let $u\in \mathcal{U}$ and assume, by way of contradiction, that
		${R^0}(u,u) \in K$.  By Lemma~\ref{pareto:l: all (u,u) together},
		for all $u\in \mathcal{U}$ ${R^0}(u,u)\in K$.  Since $K$ is
		consistent, it does not contain facts of the form ${R^0}(u,f_u)$ for
		all $u\in \mathcal{U}$ (since the fact ${R^0}(u,f_u)$ is
		inconsistent with ${R^0}(u,u)$).  Moreover, $K$ does not contain
		facts of the form ${R^0}(X_u,u)$ where $X\in \mathcal{X}$ and
		$u\in X$ (for a similar reason).  By Lemma~\ref{pareto:l:notmiddle
			R2(Xx,f_u)} (respectively, \ref{pareto:l:notmiddle R2(Xx,u)}), $K$
		does not contain facts of the form ${R^0}(X_x,f_u)$ (respectively,
		${R^0}(X_x,u)$) where $X\in \mathcal{X}$ and $u\in X$.  It holds
		that $K$ is maximal and thus it contains all of the facts of type
		$(iv)$.  Thus, we conclude that $K$ is exactly $\Jz$ which leads to
		a contradiction.
	\end{proof}
	
	We show that if a fact of type $(iv)$ is in $K$, then all of the facts of type $(iv)$ are in $K$.
	\begin{lemma}
		\label{pareto:l:whole set left column}
		Let $X\in \mathcal{X}$.  If ${R^0}(X_{x^\prime },x^\prime)\in K$ for
		some $x^\prime \in X$ then for all $x\in X$ we have
		${R^0}(X_{x},x)\in K$.
	\end{lemma}
	\begin{proof}
		Let us denote $X=\{x_0,\ldots,x_{n-1}\}$ and assume without loss of
		generality that $x_0 = x^\prime$.  We prove by induction on $i$ that
		${R^0}(X_{x_i},x_i)\in K$ for all $i=0,\ldots,n-1$.  
		\partitle{Basis} Trivial.
		\partitle{Induction Step}
		Assume ${R^0}(X_{x_i},x_i)\in K$.
		Since ${R^0}(X_{x_{i+1}},x_i)\succ {R^0}(X_{x_i},x_i)$, there must
		be a fact $f\in K$ that is inconsistent with
		${R^0}(X_{x_{i+1}},x_i)$.  Since $K$ is consistent, $f$ must agree
		with ${R^0}(X_{x_{i+1}},x_i)$ on $A$. By
		Lemma~\ref{pareto:l:notmiddle R2(Xx,f_u)}, $f$ is not of the form
		${R^0}(X_{x_{i+1}},f_u)$ for $u\in \mathcal{U}$ and by
		Lemma~\ref{pareto:l:notmiddle R2(Xx,u)} $f$ is also not of the form
		${R^0}(X_{x_{i+1}},u)$ for $u\in \mathcal{U}$. Therefore, $f$ must
		be ${R^0}(X_{x_{i+1}},x_{i+1})$.
	\end{proof}
	
	Next, we show that $K$ encodes a solution for $\XC$.  Specifically, we
	contend that there is an exact cover of $\mathcal{U}$ by
	$\mathcal{X}$, namely $\mathcal{C}$, that is defined as follows:
	$C \in \mathcal{C}$ if and only if ${R^0}(C_c,c)\in K$ for some
	$c\in C$.
	\begin{lemma}
		\label{pareto:l: all u is covered}
		For all $u\in \mathcal{U}$ there exists
		$C \in \mathcal{C}$ such that $u\in C$.
	\end{lemma}
	\begin{proof}
		Let $u\in \mathcal{U}$ and
		assume, by way of contradiction, that for all
		$C \in \mathcal{C}$, it holds that $u\not \in C$.
		By our assumption, $X\not \in \mathcal{C}$ for all $X\in \mathcal{X}$ such that $u\in X$.
		Note that since the union of the sets in $\mathcal{X}$ is $\mathcal{U}$, there exists a set $X\in \mathcal{X}$ such that $u\in X$.  
		Thus, the definition of the set $\mathcal{C}$ implies that
		for all $X\in \mathcal{X}$ such that $u\in X$, we have that ${R^0}(X_u,u)\not \in K$.
		By Lemma~\ref{pareto:l: (X_x,x) not in K, (x,x) in K}, ${R^0}(u,u) \in K$. 
		By Lemma~\ref{pareto:l: (u,u) not in K}, this is a contradiction. 
	\end{proof}
	\begin{lemma}
		\label{pareto:l: disjoint cover}
		For all 
		$C , C^\prime \in \mathcal{C}$ where $C \ne C^\prime$, it holds that 
		$C \cap C^\prime = \emptyset$.
	\end{lemma}
	\begin{proof}
		Let $C , C^\prime \in \mathcal{C}$ where $C \ne C^\prime$.
		Assume, by way of contradiction, that there exists $u\in C \cap C^\prime$.
		By $\mathcal{C}$'s definition, 
		${R^0}(C_c,c)\in K$ for some $c\in C$. 
		Lemma~\ref{pareto:l:whole set left column}
		implies that for all $c\in C$, it holds that ${R^0}(C_c,c)\in K$.
		Similarly, for all $c^\prime \in C^\prime$, we have that ${R^0}({C^\prime}_{c^\prime},c^\prime)\in K$.
		Since $u\in C \cap C^\prime$, both facts  ${R^0}(C_u,u)$ and ${R^0}({C^\prime}_u,u)$ are in $K$. This is a contradiction to $K$'s consistency.
	\end{proof}
	Finally, we conclude the following.
	\begin{lemma}
		\label{pareto:l: exact cover}
		$ \mathcal{C}$ is an exact cover of $\mathcal{U}$
	\end{lemma}
	\begin{proof}
		Follows from Lemmas~\ref{pareto:l: all u is covered} and~\ref{pareto:l: disjoint cover}.
	\end{proof}
	
	Given $K$, a p-repair different from $\Jz$, we showed that there exists a solution to $\XC$, (i.e., an exact cover of $\mathcal{U}$). This completes the ``only if'' direction.
	
	\subsubsection{Hardness of p-categoricity$\angs{\scs^6}$}
	
	We construct a reduction from CNF satisfiability  to p-categoricity$\angs{\scs^6}$. 
	The input to CNF is a formula $\psi$ with the free variables 
	$x_1,\ldots,x_n$, such that  $\psi$
	has the form $c_0 \wedge  \cdots \wedge c_m$ where each $c_j$ is a clause. Each clause is a conjunction of variables from the set 
	$\{x_i, \neg x_i : i=1,\ldots,n\}$.
	The goal is to determine 
	whether there is a true assignment
	$\tau: \{x_1,\ldots,x_n\} \rightarrow \{0,1\}$ that satisfies $\psi$. 
	Given such an input, 
	we will construct the input 
	$(I,\succ )$ for
	p-categoricity$\angs{\scs^6}$.
	For each $i=1,\ldots, n$ and $j=0,\ldots, m$, $I$ contains  the following facts:
	\begin{itemize}
		\item
		$R^6(\odot, x_i, 0)$
		\item
		$R^6(\odot,x_i,1)$
		\item
		$R^6(\otimes ,c_j,c_j)$
	\end{itemize}
	The priority relation $\succ$ is defined as follows:
	\begin{itemize}
		\item
		$R^6(\otimes, c_j,c_j) \succ R^6(\odot, x_i, 0)$ if $x_i$ appears in clause $c_j$,
		\item
		$R^6(\otimes, c_j,c_j) \succ R^6(\odot , x_i, 1)$ if $\neg x_i$ appears in clause $c_j$, and
		\item $R^6(\otimes, c_j,c_j) \succ R^6(\odot, x_i,b)$, where
		$b\in\{0,1\}$, if neither $x_i$ nor $\neg x_i$ appear in clause
		$c_j$.
	\end{itemize}
	Our construction is illustrated in Figure~\ref{fig:paretohard6}
	for the $CNF$ formula: 
	$\psi = (x_1\vee x_2 \vee \neg x_4)
	\wedge
	(\neg x_2 \vee x_3 \vee \neg x_4)
	\wedge
	(\neg x_1 \vee \neg x_3 \vee x_4)
	$.
	Observe that the subinstance 
	$J$ that consists of the facts $R^6(\otimes, c_j , c_j)$ for all $j$, is the only c-repair.
	To complete the proof, we will show that $\psi$ is satisfiable if and only if $I$ has a p-repair different from $J$.

	\begin{figure}[t]
		\centering
		\input{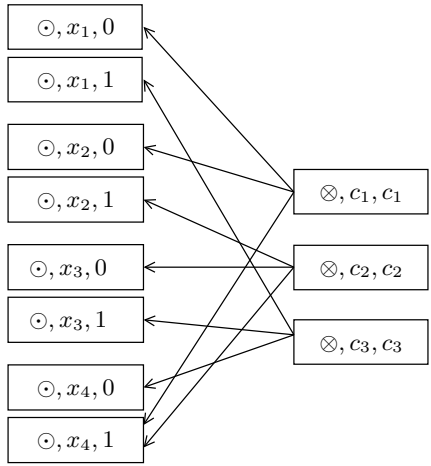}
		\caption{\label{fig:paretohard6}Illustration of the reduction 
		from $CNF$ to p-categoricity$\angs{\scs^6}$}
	\end{figure}

	\partitle{The ``if'' direction}
	Assume $\psi$ is satisfiable. That is, there exists an assignment
	$\tau:\{x_1,\ldots, x_n\}\rightarrow \{0,1\}$ that satisfies $\psi$.
	We claim that the subinstance $K$ that consists of the facts
	$R^6(\odot, x_i, \tau(x_i))$ for all $i$ is a p-repair (that is
	different from $J$).  $K$ is consistent since $\tau$ is an assignment
	(i.e., each $x_i$ has exactly one value and the constraint
	$B\rightarrow C \in \Delta^6$ is satisfied). For the same reason, $K$
	is maximal (facts of the form $R^6(\otimes, c_j, c_j)$ cannot be added
	to $K$ because of the constraint $\emptyset \rightarrow A$).  It is
	left to show that $K$ does not have a Pareto improvement. Assume, by
	way of contradiction, that it does. That is, there exists a fact
	$f \in I \setminus K$ such that $f\succ g$ for every $g\in K$. Note
	that follows from $\succ$'s definition, $f$ must be of the form
	$R^6(\otimes,c_j,c_j)$.  Nevertheless, this implies the clause $c_j$
	is not satisfied by $\tau$ which leads us to the conclusion that
	$\psi$ is not satisfied by $\tau$.
	\partitle{The ``only if'' direction}
	Assume there is a p-repair $K$ different from $J$.  Since $K$ is
	different from $J$, it must contain a fact $R^6(\odot, x_i, b_i)$ for
	some $i$.  Since $K$ is consistent and the FD
	$\emptyset \rightarrow A$ is in $\Delta^6$, it holds that for all
	facts $f$ in $K$, we have that $f[A]=\odot$.  $K$ is maximal and
	consistent and thus induces a true assignment $\tau$, defined by
	$\tau(x_i)\eqdef b_i$.  Assume, by way of contradiction, that $\tau$
	does not satisfy $\psi$. That is, there exists a clause $c_j$ that is
	not satisfied.  By $\succ$'s definition,
	$R^6(\otimes, c_j,c_j) \succ R^6(\odot, x_i, b_i)$ for each $x_i$
	that appears in $c_j$ (with or without negation). Moreover,
	$R^6(\otimes, c_j,c_j) \succ R^6(\odot,x_l, b_l)$ for each $x_l$ that
	does not appear in $c_j$ (with or without negation).  Therefore
	$R^6(\otimes, c_j,c_j) \succ f$ for every $f\in K$, which implies
	that $K$ has a Pareto improvement that contains the fact
	$R^6(\otimes, c_j,c_j)$. Hence, we get a contradiction.

	\subsection{Proof of Lemma~\ref{lemma:fw-from-s0}}
	\begin{replemma}{\ref{lemma:fw-from-s0}}
		\lemmafwfromszero
	\end{replemma}
	\begin{proof}
		Denote the single relation symbol in $\relset$ by $R$.
		It holds that $\Delta$ is equivalent to a pair of keys and therefore we can denote it by 
		$\{X \rightarrow \inds{R} , Y \rightarrow \inds{R}\}$ where $X,Y$ are subsets of $\inds{R}$. 
		Let $f = R^0(a,b)$.
		We define a fact-wise reduction ,$\Pi : R^0 \rightarrow R$, using the constant $\odot\in  \consts$. That is, denote $\Pi(f) = (d_1,\ldots,d_n)$ where for all $i=1,\ldots,n$
		\[
		d_i \eqdef
		\begin{cases}
		a & i\in \mbox{ $X \setminus Y$}\\ 
		b& i\in \mbox{$Y \setminus X$}\\
		\odot & \mbox{otherwise}.
		\end{cases}
		\] 
		In order to prove that $\Pi$ is indeed a fact-wise reduction, we should show it is well-defined, preserves consistency and inconsistency and is injective.
		
		It is straightforward to see that $\Pi$ is injective.  Note that since
		$\Delta$ is not equivalent to a single FD, it holds that
		$X\setminus Y \ne \emptyset$ and $Y\setminus X \ne
		\emptyset$.
		Therefore, there exist $i$ and $j$ such that $d_i = a$ and $d_j=b$.
		To show that it preserves consistency, we should show that for every
		two facts $f$ and $f\p$, the set $\{f,f\p\}$ is consistent w.r.t
		$\Delta^0$ if and only if $\{\Pi(f),\Pi(f\p)\}$ is consistent w.r.t
		$\Delta$.
		
		\partitle{The ``if'' direction}
		Assume that $\{f,f\p\}$  is consistent w.r.t $\Delta^0$, we contend that $\{\Pi(f),\Pi(f\p)\}$ is consistent w.r.t $\Delta$.
		Since  $\{f,f\p\}$  is consistent w.r.t $\Delta^0$, if $f$ and $f\p$ agree on $A$ then they must also agree on $B$ and if they do not agree on $A$ then they also do not agree on $B$. 
		Thus, either $a=a\p,b =b\p$ or 
		$a \ne a\p,b \ne b\p$.
		
		If $a=a\p$ and $b =b\p$ then by $\Pi$'s definition, it holds that $\Pi(f)=\Pi(f\p)$ and thus  $\{\Pi(f),\Pi(f\p)\}= \{\Pi(f)\}$ is  consistent w.r.t $\Delta$.
		
		Let $\Pi(f) = (d_1,\ldots, d_n)$ and let $\Pi(f\p) = (d\p_1,\ldots, d\p_n )$.
		If $a\ne a\p$ and $b  \ne b\p$ then by $\Pi$'s definition and since $X \setminus Y$ and $Y \setminus X$ are not empty it holds that
		there exists $i\in X \setminus Y$ such that 
		$d_i\ne  d\p_i$ and $j\in Y \setminus X$ such that  $d_j \ne d\p_j$. 
		Thus $\Pi(f)$ and $\Pi(f\p)$ do not agree on $X$ nor on $Y$. That is, $\{\Pi(f),\Pi(f\p)\}$ is consistent with respect to $\Delta $.
		
		\partitle{The ``only if'' direction}
		Assume $\{f,f\p\}$  is inconsistent w.r.t $\Delta^0$, we contend that $\{\Pi(f),\Pi(f\p)\}$ is inconsistent w.r.t $\Delta$.
		Since   $\{f,f\p\}$  is inconsistent w.r.t $\Delta^0$, if $f$ and $f\p$ agree on $A$ then they must disagree on $B$ and if they disagree on $A$ then they must agree on $B$. Thus, either $a=a\p$, $b \ne b\p$ or 
		$a \ne a\p $, $b = b\p$.
		Both cases are symmetric and thus we will prove the claim only for the case where  $a=a\p$ and $b \ne b\p$.
		Let $\Pi(f) = (d_1,\ldots,d_n)$ and $\Pi(f\p)=(d\p_1,\ldots,d\p_n)$.
		By $\Pi$'s definition,
		$d_i = d\p_i$ for $i\in X \setminus Y$ and $d_j \neq d\p_j$ for $j\in Y \setminus X$. Moreover, for all other $k$, it holds that $d_k = d\p_k = \odot$.
		That is, the facts $\Pi(f)$ and $\Pi(f\p)$ agree on $X$ and disagree on $Y$.
		Hence, $\{\Pi(f),\Pi(f\p)\}$ is inconsistent w.r.t $\Delta$
	\end{proof}

	\subsection{Proof of Lemma~\ref{lemma:fw-from-s0-to-s1-5}}
	\begin{replemma}{\ref{lemma:fw-from-s0-to-s1-5}}
		\lemmafwfromszerotosonetofive
	\end{replemma}
	\begin{proof}
		Recall our schemas
		$\scs^i$, for $i=1,\dots,5$, where each $\scs^i$ is the schema
		$(\relset^i,\depset^i)$, and $\relset^i$ is the singleton
		$\set{R^i}$. The specification of the $\scs^i$ is as follows.
		
		\begin{enumerate}
			\item[0.] $R^0/2$ and $\depset^0=\set{A\rightarrow B,B\rightarrow A}$
			\item[1.] $R^1/3$ and $\depset^1=\set{AB\rightarrow C,BC\rightarrow A,AC \rightarrow B}$
			\item[2.] $R^2/3$ and $\depset^2=\set{A\rightarrow B, B\rightarrow A}$ 
			\item[3.] $R^3/3$ and $\depset^3=\set{AB\rightarrow C, C\rightarrow B}$
			\item[4.] $R^4/3$ and $\depset^4=\set{A\rightarrow B, B\rightarrow C}$
			\item[5.] $R^5/3$ and $\depset^5=\set{A\rightarrow C , B\rightarrow C}$
		\end{enumerate} 
		
		Let $f =R^0 (a,b)$ be a fact over $\relset^0 / 2$.
		For all $i=1,\ldots, 5$ we define the fact-wise reduction $\Pi^i: R^0 \rightarrow R^i$, using  $\otimes \in \consts$, by:
		\begin{enumerate}
			\item
			$\Pi^1(f) \eqdef R^1(a,b,\otimes)$,
			\item
			$\Pi^2(f) \eqdef  R^2(a,b,\otimes )$,
			\item
			$\Pi^3(f) \eqdef  R^3(\otimes,a,b),$
			\item
			$\Pi^4(f) \eqdef  R^4(a,b,a),$
			\item
			$\Pi^5(f) \eqdef  R^5(a,b,\langle a,b \rangle ).$
		\end{enumerate}
		Regarding the fact-wise reduction $\Pi^5$, note that 
		$\langle a,b \rangle = \langle a\p,b\p \rangle$ if and only if $a=a\p$ and $b=b\p$.
		Note that for each $i=1,\ldots, 5$, we have that $\Pi^i$ is computable in polynomial time. Moreover, it is straightforward to see that for each $i=1,\ldots, 5$, it holds that $\Pi^i$ is injective.
		Thus, it is left to show that for each $i=1,\ldots, 5$, the fact-wise reduction $\Pi^i$ preserves consistency and inconsistency.
		That is, given two facts 
		$f = R^0(a,b)$ and $f\p = R^0(a\p,b\p)$, 
		we show that for each $i=1,\ldots, 5$, it holds that $\{\Pi^i(f),\Pi^i(f\p)\}$ is consistent w.r.t $\Delta^i$, if and only if $\{f,f\p\}$ is consistent w.r.t $\Delta^0$.
		
		We prove the above for $i=5$.
		\partitle{The ``if'' direction}
		Let $f = R^0(a,b)$ and $f\p = R^0(a\p,b\p)$. 
		Assume that $\{f,f\p\}$ is inconsistent with respect to $\Delta^0$. We contend that
		$\{\Pi^5(f),\Pi^5(f\p)\}$ is inconsistent w.r.t $\Delta^1$.
		Note that since $\{f,f\p\}$ is inconsistent with respect to $\Delta^0$, $f$ and $f\p$ must agree on one attribute and disagree on the second.
		Both cases are symmetric and thus it suffices to show that if  
		$f = R^0(a,b)$ and $f\p = R^0(a\p,b\p)$ agree on $A$ and disagree on $B$ then $\{\Pi^5(f),\Pi^5(f\p)\}$ is inconsistent w.r.t $\Delta^1$.
		Indeed, $f = R^0(a,b)$ and $f\p = R^0(a\p,b\p)$ agree on $A$ which implies $a=a\p$, and disagree on $B$ which implies $b\ne b\p$.
		Thus, $\Pi^5(f)=R^1(a,b,\langle a,b \rangle )$ and  
		$\Pi^5(f\p)=R^1(a\p,b\p,\langle a\p,b\p \rangle)$ agree on $A$ but do not agree on $C$. That is, $\{\Pi^5(f),\Pi^5(f\p)\}$ is inconsistent w.r.t $\Delta^1$.
		
		\partitle{The ``only if'' direction}
		Let $f = R^0(a,b)$ and $f\p = R^0(a\p,b\p)$ and assume that  $\{f,f\p\}$ is consistent with respect to $\Delta^0$.
		We contend that  $\{\Pi^5(f),\Pi^5(f\p)\}$ is consistent w.r.t $\Delta^1$.
		If
		$f$ and $f\p$ agree on $A$ (i.e., $a = a\p$), since they are consistent w.r.t $\Delta^0$ they must agree also on $B$ (i.e., $b=b\p$).
		Hence $f=f\p$ and therefore $\Pi^5(f)=\Pi^5(f\p)$.
		We conclude that $\{\Pi^5(f),\Pi^5(f\p)\}$ is consistent w.r.t $\Delta^1$.
		If
		$f$ and $f\p$ do not agree on $A$ (i.e., $a \ne a\p$), since they are consistent w.r.t $\Delta^0$ they must also disagree on $B$ (i.e., $b \ne b\p$).
		Thus,  $\Pi^5(f)$ and $\Pi^5(f\p)$ do not agree on $A$ nor on $B$. That is, $\Delta^1$ holds and thus $\{\Pi^5(f),\Pi^5(f\p)\}$ is consistent w.r.t $\Delta^1$.\\
		
		Similarly, this can be shown for each $i=1,\ldots,4$ (which are simpler cases than $i=5$).
	\end{proof}

\def\refcompletion{\ref{sec:c}}
\section{Proofs for Section~\refcompletion}

In the current section we prove the correctness of the $\algname{CCategoricity}$ algorithm, introduced in Section~\ref{sec:c}. In particular, we prove Theorem~\ref{thm:ccategoricity-correct}. For convenience, we repeat the theorem here.

\begin{reptheorem}{\ref{thm:ccategoricity-correct}}
	\theoremccategoricitycorrect
\end{reptheorem}

We will divide our proof into two parts.
First we will prove that the algorithm is sound (i.e., if $J$ is consistent, then $I$ has precisely one c-repair).
Later, we will prove that the algorithm is complete (i.e., if $I$ has precisely one c-repair, then $J$ is consistent).
Before that, we need a basic lemma that will be used in both parts of the proof.

\subsection{Basic Lemma}

Let $(I,\graph,\succ)$ be an inconsistent prioritizing instance over a signature $\relset$. We start by proving the following lemma.

\begin{lemma}\label{lemma:ccat-findcrepair-from-consistent-j}
	Suppose that $P_1,\dots,P_t$ are the positive strata constructed by executing $\algname{CCategoricity}(I,\graph,\succ)$. For every $k\in\set{1,\dots,t}$, if $J_k=P_1\cup\dots\cup P_k$ is consistent, then there exists an execution of the $\algname{FindCRep}$ algorithm on $(I,\graph,\succ)$, such that at the beginning of some iteration of that execution the following hold:
	\begin{itemize}
		\item The set of facts included in $J$ is $P_1\cup\dots\cup P_k$.
		\item Every fact $f\in P_{k+1}$ belongs to $\max_\succ(I)$.
	\end{itemize}
\end{lemma}

\begin{proof}
	Let us start the execution of the $\algname{FindCRep}$ algorithm on $(I,\graph,\succ)$ as follows. First, we will select all of the facts from $P_1$ and add them to $J$ (one by one), then we will select all of the facts from $P_2$ and add them to $J$ and so on. If we can add all of the facts from $P_1\cup\dots\cup P_k$ to $J$ during this process, and at the end all of the facts in $P_{k+1}$ belong to $\max_\succ(I)$, then there exists an execution of the $\algname{FindCRep}$ algorithm that satisfies all the conditions and that will conclude our proof.
	
	Otherwise, one of the following holds.
	\begin{itemize}
		\item One of the facts in $P_1\cup\dots\cup P_k$ cannot be added to $J$.
		\item After adding all of the facts from $P_1\cup\dots\cup P_k$ (and only these facts) to $J$, one of the facts in $P_{k+1}$ does not belong to $\max_\succ(I)$.
	\end{itemize}
	If the first case holds, let $P_i$ be the first positive stratum such that a fact $f\in P_i$ cannot be added to $J$ after adding all of the facts in $P_1\cup\dots\cup P_{i-1}$ and maybe some of the facts in $P_i$ to $J$. Note that $P_1\cup\dots\cup P_k$ does not contain a hyperedge, thus if $f$ cannot be added to $J$, it holds that $f$ does not belong to $\max_\succ(I)$. That is, there exists at least one fact $g$, such that $g\succ f$. Moreover, there exists at least one finite sequence, $h_1\succ\dots\succ h_m\succ f$, of facts in $I$ (with $m\ge 1$), such that $h_1\in \max_\succ(I)$. Note that for each such sequence of facts, it holds that $h\succs f$ for every fact $h$ in the sequence.
	In this case, one of the following holds for each fact $h$:
	\begin{itemize}
		\item $h$ belongs to $P_j$ or $N_j$ for some $j<i$.
		\item $h$ belongs to $P_j$ or $N_j$ for some $j\ge i$.
	\end{itemize}
	We assumed that we were able to add all the facts from $P_1\cup\dots\cup P_{i-1}$ to $J$, thus $h\in\set{h_1,\dots,h_m}$ cannot belong to $P_j$ for some $j<i$. Furthermore, by the definition of $P_i$, after removing all of the facts from $P_1\cup\dots\cup P_{i-1}$ and $N_1\cup\dots\cup N_{i-1}$, every fact $f\in P_i$ belongs to $\max_\succs(I)$ (and consequently to $\max_\succ(I)$). This cannot be the case if there exists a fact $h\succs f$ that belongs to $P_j$ or $N_j$ for some $j\ge i$. Thus, the only possibility left is that every fact $h\in\set{h_1,\dots,h_m}$ belongs to $N_j$ for some $j<i$.
	
	Since for every such sequence of facts, $h_1$ belongs to $\max_\succ(I)$, it can be selected in line~3 of the algorithm at the next iteration. By the definition of $N_j$, there exists a hyperedge that is contained in $P_1\cup\dots\cup P_j\cup\set{h_1}$. We know that all the facts in $P_1\cup\dots\cup P_j$ were added to $J$, thus $h_1$ will be excluded from $J$. After selecting fact $h_1$ from each sequence and removing it from $I$, the next fact, $h_2$, in each sequence belongs to $\max_\succ(I)$. The previous arguments hold for $h_2$ as well, thus $h_2$ can be selected  in line~3 of the algorithm at the next iteration. Note that if a fact $h$ belongs to more than one sequence, it will be removed after all the previous facts in each one of these sequences are removed. We can continue with this process until it holds that $f\in \max_\succ(I)$, and then add it to $J$, in contradiction to our assumption. 
	
	If the second case holds, at each of the next iterations of the algorithm (and before adding facts that do not belong to $P_1\cup\dots\cup P_k$ to $J$), there exists a fact $f\in P_{k+1}$ that does not belong to $\max_\succ(I)$. Similarly to the previous part, if $f$ does not belong to $\max_\succ(I)$, there exist at least one sequence $h_1\succ\dots\succ h_m\succ f$ of facts, such that $h\succs f$ for every fact in the sequence and each fact $h$ belongs to $N_j$ for some $j<{k+1}$. We again can select all of the facts in each sequence as the maximal elements in line~3 of the algorithm by topological order until $f\in \max_\succ(I)$. This is a contradiction to our assumption, thus $f$ does belong to $\max_\succ(I)$ after adding all of the facts in $P_1\cup\dots\cup P_k$ to $J$ and removing all of the corresponding facts in $N_1\cup\dots\cup N_k$ from $I$. Note that in this case, we will be able to select $f$ in line~3 of the algorithm at the next iteration; however, we may not be able to add $f$ to $J$, since $J_{k+1}$ is not consistent.
\end{proof}

\begin{example}
	Consider the $\algname{CCategoricity}$ execution on the inconsistent prioritizing instance $(I,\graph,\succ)$ from our followers running example, illustrated in Figure~\ref{fig:follows-exec}. In this example, $P_1\cup P_2$ is consistent. We can start building the repair $J$, using the $\algname{FindCRep}$ algorithm as follows. (We recall that we denote inclusion in $J$ by plus and exclusion from $J$ by minus.)
	\[ +f_{11},+f_{34},-f_{12},-f_{21},-f_{31},+f_{22},+f_{35}\]
	As Lemma~\ref{lemma:ccat-findcrepair-from-consistent-j} states, at the current iteration of the algorithm, all the facts in $P_1\cup P_2$ belong to $J$. Moreover, the facts in $P_3$ (that is, $f_{23}$ and $f_{32}$) belong to $\max_\succ(I)$, as the second part of the lemma states.
	
	Note that in our example, there exists a different execution of the $\algname{FindCRep}$ algorithm.
	\begin{gather*}
		+f_{11},-f_{12},-f_{21},-f_{31},+f_{22},\\
		+f_{23},+f_{32},-f_{24},+f_{34},+f_{35}
	\end{gather*}
	In this case, there does not exist an iteration of the algorithm that satisfies the conditions from Lemma~\ref{lemma:ccat-findcrepair-from-consistent-j}. Thus, different executions of the algorithm may be possible; however, there always exists an execution of the $\algname{FindCRep}$ algorithm that satisfies the conditions from the lemma.
\end{example}

\subsection{Soundness}

Next, we prove that if $J$ is consistent, then $I$ has precisely one c-repair. Let us denote by $t$ the number of iterations of the $\algname{CCategoricity}$ algorithm on the input $(I,\graph,\succ)$. That is, it holds that $J=P_1\cup\dots\cup P_t$ at the end of the algorithm. Since $J$ is consistent, Lemma~\ref{lemma:ccat-findcrepair-from-consistent-j} and the fact that the $\algname{FindCRep}$ algorithm is sound (Theorem~\ref{thm:cgreedy}) imply that there exists a c-repair $K$ that includes all of the facts in $P_1\cup\dots\cup P_t$. The following lemma proves that a c-repair cannot include any fact that belongs to $N_i$ for some $i\in\set{1,\dots\,t}$.

\begin{lemma}\label{lemma:ccat-no-repair-Ni}
	If $J$ is consistent, then no c-repair contains any fact from $N_1\cup\dots\cup N_t$.
\end{lemma}

\begin{proof}
	Let us assume, by way of contradiction, that there exists a c-repair $K'$ that includes at least one fact from $N_1\cup\dots\cup N_t$. By Theorem~\ref{thm:cgreedy}, $\algname{FindCRep}$ is complete, thus there exists an execution of the algorithm that produces $K'$. Consider an execution of $\algname{FindCRep}$ that produces $K'$. In that execution, consider the first time that a fact from a negative stratum is added to $K'$; let $g$ be that fact. That is, $g\in N_i$ for some $i\in\set{1,\dots,t}$. By the definition of $N_i$, there exists a hyperedge $e=\set{f_1,\dots,f_m,g}$ that is contained in $P_1\cup\dots\cup P_i\cup\set{g}$, such that for every other fact $f\in e$ it holds that $f\succs g$. If we are able to choose fact $g$ for the repair at some iteration of the algorithm, it necessarily belongs to $\max_\succ(I)$ at this iteration. Therefore, all of the other facts in $e$ are no longer included in $I$. 
	
	This may be the case if all of the other facts in the hyperedge were already added to $K'$. However, adding $g$ to $K'$ will result in a hyperedge, in contradiction to the fact the $K'$ is a repair. Therefore, at least one of the facts in the hyperedge was removed from $I$ without being added to $K'$; let $f$ be such a fact. That is, there exists another hyperedge $e'$ that includes $f$, such that all of the other facts in this hyperedge were added to $K'$ before $f$ was removed from $I$. All of these facts, including $f$, belong to $P_1\cup\dots\cup P_t$, since we assumed that $g$ is the first fact from some $N_i$ that was chosen for the repair. Hence, we found a hyperedge, $e'$, that is contained in $P_1\cup\dots\cup P_t$, in contradiction to the fact that $P_1\cup\dots\cup P_t$ is consistent. Thus, a fact $g\in N_i$ cannot be added to any repair of $I$.
\end{proof}

Lemma~\ref{lemma:ccat-no-repair-Ni} implies that every c-repair is contained in $J$. Moreover, as said above, there exists a c-repair that includes all of the facts in $J$. Thus, $J$ is the only c-repair of $I$ and this concludes the proof of soundness of the $\algname{CCategoricity}$ algorithm.

\subsection{Completeness}

Finally, we prove that if $I$ has precisely one c-repair, then $J$ is consistent. Throughout this section we fix $(I,\graph,\succ)$ and assume that there is exactly one c-repair, which we denote by $K$. Let $J_i$ denote the subinstance $P_1\cup\dots\cup P_i$. We prove by induction on $k$ that after the $k$th iteration of $\algname{CCategoricity}$, the instance $J_k=P_1\cup\dots \cup P_k$ is consistent. Then, we will conclude that $J$ is consistent when the algorithm reaches line~9, and consequently, the algorithm returns true as we expect. 

The basis of the induction, $k=0$, is proved by observing that $J_0$ is an empty set, thus it does not include any hyperedge. For the inductive step, we need to prove that if $J_k$ is consistent then $J_{k+1}$ is also consistent. So, suppose that $J_k$ is consistent. Let us assume, by way of contradiction, that $J_{k+1}$ is inconsistent, that is, $J_{k+1}$ contains a hyperedge. We next prove the following lemma.

\begin{lemma}\label{lemma:ccat-next-p-no-conflicts}
	There exists a fact $f\in P_{k+1}$ and a hyperedge $e$ such that the following hold:
	\begin{itemize}
		\item $f\in e$.
		\item $(e\setminus\set{f})\subseteq P_1\cup\dots\cup P_k$.
	\end{itemize}
\end{lemma}

\begin{proof}
	Let $e$ be a hyperedge that is contained in $P_1\cup\dots\cup P_{k+1}$ and has a minimal intersection with $P_{k+1}$. Let $\set{f_1,\dots,f_m}$ be the set $e\cap P_{k+1}$. Then $m>0$ since $J_k$ does not contain a hyperedge. To prove the lemma, we need to show that $m=1$. Suppose, by way of contradiction, that $m>1$. Since $J_k$ is consistent, Lemma~\ref{lemma:ccat-findcrepair-from-consistent-j} implies that we can start building a c-repair using the $\algname{FindCRep}$ algorithm by first choosing all of the facts in $P_1\cup\dots\cup P_k$. Moreover, after choosing all of these facts, there exists an iteration $i$ in which each fact in $P_{k+1}$ belongs to $\max_\succ(I)$. Since $\algname{FindCRep}$ always produces a c-repair (Theorem~\ref{thm:cgreedy}), we can now choose the fact $f_1$, which will result in a c-repair $J_1$. This holds true due to our assumption that $m>1$ and $m$ is minimal (hence, adding a signle fact to a set of facts that currently includes only facts from $P_1\cup\dots\cup P_k$ does not result in the containment of a hyperedge).
	
	The c-repair $J_1$ cannot contain all of the facts in $\set{f_1,\dots,f_m}$, since a repair cannot contain a hyperedge. Let us assume that $f_j\in\set{f_1,\dots,f_m}$ is not in $J_1$. If we choose $f_j$ instead of $f_1$ at the $i$th iteration, the $\algname{FindCRep}$ algorithm will produce a c-repair $J_2$ that includes $f_j$. Again, we can choose $f_j$ due to our assumption that $m>1$ and $m$ is minimal. That is, we have two distinct c-repairs, $J_1$ and $J_2$, in contradiction to our assumption that $I$ has precisely one c-repair.
\end{proof}

Since we assumed that $I$ has exactly one c-repair and $P_1\cup\dots\cup P_k$ is consistent, Lemma~\ref{lemma:ccat-next-p-no-conflicts} implies that there exists a hyperedge $e=\set{f_1,\dots,f_m}$ such that precisely one of the $f_i$ belongs to $P_{k+1}$, while the other facts belong to $P_1\cup\dots\cup P_k$. Without loss of generality, we can assume that $f_1\in P_{k+1}$. We next prove the existence of two distinct c-repairs:
\begin{itemize}
	\item A c-repair that includes all of the facts in $e\setminus\set{f_1}$ and does not include $f_1$. 
	\item A c-repair that includes $f_1$. 
\end{itemize}

\begin{lemma}\label{lemma:ccat-crepair-no-f1}
	There exists a c-repair that does not include $f_1$.
\end{lemma}

\begin{proof}
	Lemma~\ref{lemma:ccat-findcrepair-from-consistent-j} implies that it is possible to build a c-repair using $\algname{FindCRep}$ by first choosing all of the facts in $P_1\cup\dots\cup P_k$. As all of the facts in $\set{f_2,\dots,f_m}$ belong to $P_1\cup\dots\cup P_k$, all of them will be chosen in this process as well. Since the $\algname{FindCRep}$ algorithm is sound (Theorem~\ref{thm:cgreedy}), this specific execution of the algorithm will result in a c-repair $J$ of $I$ that includes all of the facts $f_2,\dots,f_m$. The fact $f_1$ cannot be included in $J$, since adding it to $J$ will result in a repair that contains a hyperedge, which is impossible by definition.
\end{proof}

To complete the proof of completeness, we prove that there exists another c-repair of $I$ that includes $f_1$. In order to do so, we again take advantage of the algorithm $\algname{FindCRep}$, and prove the following lemma.

\begin{lemma}\label{lemma:ccat-crepair-f1}
	There exists a c-repair that includes $f_1$.
\end{lemma}

\begin{proof}
	In order to prove the lemma, we start building the corresponding c-repair, using the $\algname{FindCRep}$ algorithm, by first selecting in line~3 of the algorithm only facts from $P_1\cup\dots\cup P_k$ that either cannot be added to $J$ or satisfy at least one of the following conditions:
	\begin{itemize}
		\item $g\succs f_1$.
		\item $f_1$ and $g$ are not neighbors in $\graph$.
	\end{itemize}
	Note that not all of the facts in $P_1\cup\dots\cup P_k$ that satisfy at least one of the conditions can be selected in this process. A fact $g\in P_2$, for example, may be left out of $J$ if a fact $f\in P_1$ is not selected because it does not satisfy any of the conditions and it holds that $f\succ g$. In this case, $g$ will not belong to $\max_\succ(I)$ until $f$ is selected by the algorithm. However, since $J_k$ is consistent, all of the facts from $P_1\cup\dots\cup P_k$ that can be selected in line~3 of the algorithm during this process can also be added to $J$. 
	
	We will now prove that after selecting all of these facts, the algorithm can add the fact $f_1$ to $J$ next (that is, before adding any other fact to $J$). This will eventually result in a c-repair that includes $f_1$ and will conclude our proof. Let us assume, by way of contradiction, that we cannot add $f_1$ to $J$ next. That is, one of the following holds:
	\begin{itemize}
		\item There exists a hyperedge that contains $f_1$ such that all of the other facts in the hyperedge have already been added to $J$.
		\item We must add another fact to $J$ before it holds that $f_1\in\max_\succ(I)$.
	\end{itemize}
	
	If the first case holds, then there exists a hyperedge $e$, such that $f_1\in e$ and all of the other facts in the hyperedge were added to $J$ in the previous iterations of the algorithm. Note that all of the facts that have already been added to $J$, including $e\setminus\set{f_1}$, belong to $P_1\cup\dots\cup P_k$. Moreover, it holds that $h\succs f_1$ for every fact $h\in e\setminus\set{f_1}$, since a fact $h\in e\setminus\set{f_1}$ that does not hold $h\succs f_1$, does not satisfy any of the conditions and could not have been added to $J$. Hence, we found a hyperedge $e$, such that $f_1\in e$, and for every other fact $h\in e$ it holds that $h\in P_1\cup\dots\cup P_k$ and $h\succs f_1$. By the definition of negative stratum, it should hold that $f_1\in N_i$ for some $i\in\set{1,\dots\,{k+1}}$ (the exact value of $i$ depends on the other facts in this hyperedge), in contradiction to the fact that $f_1\in P_{k+1}$. Thus, the first case is impossible.
	
	If the second case holds, then at one of the next iterations, only facts that can be added to $J$ belong to $\max_\succ(I)$, while $f_1$ does not belong to $\max_\succ(I)$ yet. In this case, the $\algname{FindCRep}$ algorithm must add another fact to $J$ before adding $f_1$ to $J$. Since $f_1$ does not belong to $\max_\succ(I)$, there exists at least one fact $g$ such that $g\succ f_1$. Moreover, there exists at least one finite sequence, $h_1\succ\dots\succ h_m\succ f$, of facts in $I$ (with $m\ge 1$), such that $h_1\in \max_\succ(I)$. Note that for each such sequence of facts, it holds that $h\succs f$ for every fact $h$ in the sequence.
	In this case, one of the following holds for each fact $h$:
	\begin{itemize}
		\item $h$ belongs to $P_j$ or $N_j$ for some $j<{k+1}$.
		\item $h$ belongs to $P_j$ or $N_j$ for some $j\ge {k+1}$.
	\end{itemize}
	By the definition of $P_{k+1}$, after removing all of the facts from $P_1\cup\dots\cup P_k$ and $N_1\cup\dots\cup N_k$, every fact $f\in P_{k+1}$ belongs to $\max_\succs(I)$. This cannot be the case if there exists a fact $h\succs f$ that belongs to $P_j$ or $N_j$ for some $j\ge {k+1}$. Thus, each fact $h$ either belongs to $P_j$ or $N_j$ for some $j<{k+1}$. If $h_1$ belongs to some $P_j$, then $h_1$ was not added to $J$ since it does not satisfy any of the conditions. This holds true since $h_1\in \max_\succ(I)$, thus it can be selected in line~3 of the algorithm. That is, there exists a hyperedge $e\subseteq P_1\cup\dots\cup P_{k+1}$, such that $\set{f_1,h_1}\subseteq e$ and it does not hold that $h_1\succs f_1$. This is a contradiction to the fact that $h_1\succs f$, thus this cannot be the case. 
	
	The only possibility left is that $h_1\in N_j$ for some $j<{k+1}$. Since it also holds that $h_1\in\max_\succ(I)$, the fact $h_1$ can be selected in line~3 of the algorithm next. Note that by the definition of $N_j$, there exists a hyperedge that is contained in $P_1\cup\dots\cup P_j\cup\set{h_1}$, such that for every other fact $h'$ in the hyperedge it holds that $h'\succs h_1$. Since $h_1\in\max_\succ(I)$ none of these facts still belongs to $I$, thus they have already been added to $J$, and $h_1$ will not be added to $J$, but only removed from $I$. We can continue with this process and choose all of the facts in each one of the corrsponding sequences by topological order, until $f_1\in \max_\succ(I)$. (If a fact $h$ belongs to more than one sequence, it will be removed after all the previous facts in each one of these sequences are removed.) 
	
	Then, we can add $f_1$ to $J$, since every hyperedge that contains $f_1$ and is included in $P_1\cup\dots\cup P_k\cup\set{f_1}$ also contains at least one fact $g$ that does not satisfy $g\succs f$ (otherwise, $f_1$ would belong to some $N_j$). That is, each hyperedge that contains $f_1$ contains at least one fact that does not satisfy any of the conditions and was not added to $J$, and adding $f_1$ to $J$ will not close any hyperedge. This is a contradiction to our assumption, thus $f_1$ can be chosen for the repair next. As said above, this concludes our proof.
\end{proof}

From Lemma~\ref{lemma:ccat-crepair-no-f1} and Lemma~\ref{lemma:ccat-crepair-f1} we can conclude that there exist two distinct c-repairs of $I$. This is a contradiction to our assumption that $I$ has precisely one c-repair. Hence, $J_{k+1}$ is necessarily consistent, and this concludes our proof of completeness.
\def\refglobal{\ref{sec:g}}
\newcommand{\gcp}{g-categoricity   problem }
\section{Proofs for Section~\refglobal}

\subsection{Proof of Theorem~\ref {thm:piptwo}}

\begin{reptheorem}{\ref{thm:piptwo}}
	\thmpiptwo
\end{reptheorem}
\begin{proof}
	To show $\piptwo$-hardness, we construct a reduction from $\qsattwo$
	to the problem of g-categoricity$\angs{\scs^6}$.  The input to
	$\qsattwo$ consists of a CNF formula $\psi(\tup{x},\tup{y})$ where
	$\tup{x}$ and $\tup{y}$ are disjoint sequences of variables.  The
	goal is to determine whether for every assignment to $\tup{x}$ there
	exists an assignment to $\tup{y}$ such that the two satisfy $\psi$.
	We denote $\tup{x}$ by $x_1,\ldots, x_n$, $\tup{y}$ by
	$y_1,\ldots, y_k$ and $\psi = c_1 \wedge \cdots \wedge c_m$.  The
	input $(I,\succ)$ is constructed from $\psi$ by adding to it the
	following facts:
	\begin{itemize}
		\item
		$R^6(0,x_i, 0)$ and $R^6(0,x_i,1)$ for each variable $x_i$,
		\item
		$R^6(1,x_i, 0)$ and $R^6(1, x_i,1)$ for each variable $x_i$,
		\item
		$R^6(1, y_l, 0)$ and $R^6(1, y_l,1)$ for each variable $y_l$,
		\item
		$R^6(0,c_j,c_j)$ for each clause $c_j$  
		\item
		$R^6(2, 0,0)$ which we denote by $f_0$.
	\end{itemize}
	The priority is defined by:
	\begin{itemize}
		\item
		$R^6(1,x_i, b) \succ R^6(0,x_i, b)$ for all $b=0,1$ and for all variables $x_i$,
		\item
		$R^6(1,x_i, 1)\succ R^6(0,c_j, c_j)$ if $x_i$ appears in clause $c_j$,
		\item
		$R^6(1,x_i, 0)\succ R^6(0,c_j, c_j)$ if $\neg x_i$ appears in clause $c_j$ and
		\item
		$f_0\succ R^6(1,w,b)$ for all variables $w$ and $b=0,1$.
	\end{itemize}

	\begin{figure}[t]
		\centering
		\input{globalPip2.pspdftex}
		\caption{\label{fig:globalPip2}Illustration of the reduction from $\qsattwo$ to g-categoricity$\angs{\scs^6}$}
	\end{figure}

	Our construction is illustrated in Figure~\ref{fig:globalPip2}
	for
	the following input to  
	$\qsattwo$: 
	$\tup{x} = x,w$,
	$\tup{y}=y,z$
	and 
	$\psi = (w\vee y \vee z)
	\wedge
	(x \vee \neg w \vee y)
	\wedge
	(\neg x \vee \neg w \vee \neg z)
	$.
	
	Observe that the subinstance $\{f_0\}$ is a c-repair and thus it is
	also a g-repair.  To complete the proof we will show that $\{f_0\}$
	is the only g-repair if and only if $\psi$ is a ``yes'' instance;
	that is, for every assignment to $\tup x$ there exists an assignment
	to $\tup y$ such that the two satisfy $\psi$.

	\partitle{The ``if'' direction:}
	Assume that $\psi$ is a ``yes'' instance and let $J$ be a g-repair of
	$I$. We contend that $J$ is exactly $\{f_0\}$.  Assume, by way of
	contradiction that $J\ne \{f_0\}$.  Since $J$ is consistent and since
	$\Delta^6$ contains the FD $\emptyset \rightarrow A$, it holds that
	all of the facts in $J$ agree on $A$.  By our assumption that
	$J\ne \{f_0\}$, there are two cases:
	\begin{enumerate}
		\item 
		For each fact $f\in J$, it holds that $f[A]=1$. 
		By $\succ$'s definition, we have that $\{f_0\}$ is a global improvement of $J$ which is in contradiction to $J$ being a g-repair.
		\item For each fact $f\in J$, it holds that $f[A] = 0$.  We state that
		for all $j$, the fact $R^6(0,c_j,c_j) \in J$.  Assume, by way of
		contradiction, there exists $j$ for which
		$R^6(0,c_j,c_j)\not \in J$.  Since $J$ is maximal it contains a fact
		$f$ that is inconsistent with $R^6(0,c_j,c_j)$.  Since for each fact
		$f\in J$, it holds that $f[A] = 0$, there is no such fact and as a
		consequence no such $j$.  Hence for all $j$, we have that
		$R^6(0,c_j,c_j) \in J$.  Moreover, since $J$ is maximal, it must
		contain a fact $R^6(0,x_i,b)$ for each $x_i \in \tup{x}$.  Note that
		$\Delta^6$ contains the FD $B\rightarrow C $. This insures that $J$
		cannot contain both of the facts $R^6(0,x_i,0)$ and $R^6(0,x_i,1)$.
		This implies that $J$ encodes an assignment $\tau_{\tup{x}}$ for the
		variables $x_i\in \tup{x}$.  This assignment is given by
		\[
		\tau_{\tup{x}}(x_i) \eqdef
		\begin{cases}
		0 & \mbox{$R^6(0,x_i,0)\in J$}\\ 
		1& \mbox{$R^6(0,x_i,1)\in J$}
		\end{cases}
		\] 
		Since $\psi$ is a ``yes'' instance, there exists an assignment
		$\tau_{\tup{y}}$ for the variables $y_l\in \tup{y}$ that together with
		$\tau_{\tup{x}}$ satisfies $\psi$.  Let $K$ be the subinstance of $I$
		that consists of the facts $R^6(1,x_i,\tau_{\tup{x}}(x_i))$ for all
		$i$ and $R^6(1,y_l,\tau_{\tup{y}}(y_l))$ for all $l$. Note that $K$ is
		consistent.  We contend that $K$ is a global improvement of $J$.
		Since $J$ and $K$ are disjoint, it suffices to show that for every
		fact $f\in J$ there is a fact $f^\prime \in K$ such that
		$f^\prime \succ f$.  Let $f\in J$. If $f$ is of the form
		$R^6(0,x_i,b)$ then we choose $f^\prime = R^6(1,x_i,b)$.  If $f$ is of
		the form $(0,c_j,c_j)$, then we choose $f^\prime=R^6(1,w,b)$ where $w$
		is the variable of a literal that satisfies $c_j$ under the union of
		$\tau_{\tup{x}}$ and $\tau_{\tup{y}}$.  We conclude that $K$ is a
		global improvement of $J$, That is, $J$ is not a g-repair in
		contradiction to our assumption.
	\end{enumerate}

	\partitle{The ``only if'' direction:}
	Assume that $\{f_0\}$ is the only g-repair of $I$. We contend that for
	every assignment to $\tup{x}$ there exists an assignment to $\tup{y}$
	such that the two satisfy $\psi$.  Let $\tau_{\tup{x}}$ be an
	assignment for $\tup{x}$ and let $J$ be a consistent subinstance of
	$I$ that consists of the facts $R^6(0,c_j,c_j)$ for $ j=1,\ldots,k$
	and $R^6(0,x_i,\tau_{\tup{x}}(x_i))$ where $ x_i\in \tup{x}$.  Since
	$\{f_0\}$ is the only g-repair, it holds that $J$ has a global
	improvement. Let us denote such a global improvement by $K$.  Assume,
	without loss of generality, that $K$ is maximal (if not, we can extend
	$K$ with additional facts).  By $\succ$'s definition, $K$ must consist
	of facts of the form $R^6(1,w,b)$.  Since the FD
	$\emptyset \rightarrow A$ is in $\Delta^6$, for each fact $f\in K$, it
	holds that $f[A]=1$. Since $K$ is maximal and since $B \rightarrow C$
	is in $\Delta^6$, we have that $K$ encodes an assignment $\tau$ for
	$\tup{x}$ and $\tup{y}$.  This assignment is given by
	\[
	\tau(w) \eqdef
	\begin{cases}
	0 & \mbox{$R^6(1,w,0)\in K$}\\ 
	1& \mbox{$R^6(1,w,1)\in K$}
	\end{cases}
	\] 
	Since $K$ is a global improvement of $J$, it must contain the fact
	$R^6(1,x_i,b_i)$ whenever $R^6(0,x_i,b_i)$ is in $J$ (this is true
	since no other fact has a priority over $R^6(0,x_i,b_i)$).  Therefore
	$\tau$ extends $\tau_{\tup{x}}$.  Finally we observe that every $c_j$
	is satisfied by $\tau$. A satisfying literal is one that corresponds
	to a fact $R^6(1,w,b)$ that satisfies $R^6(1,w,b)\succ c_j$.  We
	conclude that $\psi$ is a ``yes'' instance as claimed.
\end{proof}

\subsection{Proof of Theorem~\ref {thm:piptwoextend}}

\begin{reptheorem}{\ref{thm:piptwoextend}}
	\thmpiptwoextend
\end{reptheorem}
\begin{proof}
	Recall the schema $\scs^6 = (\relset, \depset)$ where $\relset$ consists of a single ternary relation $R^6/3$ and $\Delta^6 = \{\emptyset \rightarrow A, B\rightarrow C\}$.
	We define a fact-wise reduction $\Pi:R^6 \rightarrow R$, using the constants $\odot, \oplus \in \consts$.
	Let $f = R^6(a,b,c)$.
	We define $\Pi $ by
	$\Pi (f) = R^6(d_1,\ldots, d_{n}) $ where for all $i=1,\ldots,n$
	\[
	d_i \eqdef
	\begin{cases}
	\odot & \mbox{$i\in X$} \\ 
	a& \mbox{$i\in Y \setminus X$}\\ 
	b& \mbox{$i\in W \setminus (X\cup Y\cup Z)$} \\
	c& \mbox{$i\in Z \setminus (X\cup Y\cup W)$}\\
	\oplus& \mbox{otherwise}
	\end{cases}
	\]
	It is left to show that $\Pi$ is a fact-wise reduction.
	To do so, we prove that $\Pi$ is well defined, is injective and preserves consistency and inconsistency.
	
	\partitle{$\mathbf{\Pi}$ is well defined}
	It suffices to show that each $d_i$ is well-defined.
	We show that the sets in the definition of $d_i$ are pairwise disjoint.
	Indeed, $X$ is disjoint from the sets $Y\setminus X $, 
	$W \setminus  (X\cup Y\cup Z)$ and $Z \setminus (X\cup Y\cup W)$.
	Moreover, $Y \setminus X$ is a subset of $Y$ and therefore it is disjoint from $W \setminus  (X\cup Y\cup Z)$ and $Z \setminus (X\cup Y\cup W)$.
	Clearly,  $W \setminus  (X\cup Y\cup Z)$ and $Z \setminus (X\cup Y\cup W)$ are disjoint.
	Hence, each $d_i$ is well-defined.
	
	\partitle{$\mathbf{\Pi}$ is injective}
	Let $f,f\p \in R^6$ where $f=R^6(a,b,c)$ and $f\p = R^6(a\p,b\p,c\p)$.
	Assume that $\Pi (f) = \Pi (f\p)$.  Let us denote
	$\Pi (f)=(d_1,\ldots, d_n)$ and $\Pi (f\p)=(d\p_1,\ldots, d\p_n)$.
	Note that $Y\setminus X $ is not empty since the FD $X\rightarrow Y$
	is nontrivial.  Moreover, $W \setminus (X\cup Y\cup Z)$ and
	$Z \setminus (X\cup Y\cup W)$ are not empty since each of $W$ and $Z$
	contains an attribute that is in none of the other three sets.
	Therefore, there are $i$, $j$ and $k$ such that $d_i = a$, $d_j = b$
	and $d_k = c$.  Hence, $\Pi (f) = \Pi (f\p)$ implies that
	$d_i = d\p_i$, $d_j = d\p_j$ and $d_k = d\p_k$. Therefore we obtain
	$a=a\p$, $b=b\p$ and $c=c\p$ which implies $f=f\p$.
	
	\partitle{$\mathbf{\Pi}$ preserves consistency}
	Let $f=R^6(a,b,c)$ and $f\p = R^6(a\p,b\p,c\p)$.
	We contend that  $\{f,f\p\}$ is consistent w.r.t 
	$\Delta^6$ if and only if  $\{\Pi(f),\Pi(f\p)\}$ is consistent w.r.t $\Delta$.
	\paragraph{The ``if'' direction:}
	Assume $\{f,f\p\}$ is consistent w.r.t 
	$\Delta^6$. We prove that  $\{\Pi(f),\Pi(f\p)\}$ is consistent w.r.t $\Delta$.
	Note that $\Pi(f)$ and $\Pi(f\p)$ agree on $X$ since for each $i\in X$ we have that $d_i =\odot$, regardless of the input. Since $\{f,f\p\}$ is consistent w.r.t 
	$\Delta^6$, it holds that $a=a\p $.
	By the definition of $\Pi$ and since $a=a\p $, we have that  $\Pi(f)$ and $\Pi(f\p)$ agree on $Y$. Hence,  $\{\Pi(f),\Pi(f\p)\}$ satisfies the constraint $X\rightarrow Y$.
	Assume that $\Pi(f)$ and $\Pi(f\p)$ agree on $W$.
	By the definition of $\Pi$ , since $W \setminus (X\cup Y \cup Z)$ is not empty, it holds that $b=b\p $.
	Since  $\{f,f\p\}$ is consistent w.r.t 
	$\Delta^6$, the fact that  $b=b\p $ implies that also $c=c\p $ (due to the constraint $B\rightarrow C$).
	Hence,  $\Pi(f)=\Pi(f\p)$. This implies and that $\{\Pi(f),\Pi(f\p)\} = \{\Pi(f)\}$ is consistent w.r.t $\Delta$.
	\paragraph{The ``only if'' direction:}
	Assume $\{f,f\p\}$ is inconsistent w.r.t 
	$\Delta^6$. We prove that  $\{\Pi(f),\Pi(f\p)\}$ is inconsistent w.r.t $\Delta$.
	There are two cases
	\begin{itemize}
		\item
		$f$ and $f\p $ do not agree on $A$.
		It holds, by $\Pi$'s definition, that $\Pi(f)$ and $\Pi(f\p )$ agree on $X$. 
		Nevertheless, since $f$ and $f\p$ do not agree on $A$, we have that $a\ne a\p $.
		Hence $\Pi(f)$ and $\Pi(f\p )$ do agree on $Y$. That is,  the constraint $X\rightarrow Y$ is not satisfied, which leads us to the conclusion that  $\{\Pi(f),\Pi(f\p)\}$ is inconsistent w.r.t $\Delta$.  
		\item
		$f$ and $f\p $ agree on $A$.
		Since  $\{f,f\p\}$ is inconsistent w.r.t 
		$\Delta^6$, we have that $f$ and $f\p $ agree on $B$ (i.e., $b=b\p $) but disagree on $C$ (i.e., $c=c\p $).
		Note that  $\Pi(f)$ and $\Pi(f\p )$ agree on $W$ since $a=a\p $ and $b =b\p $. Nevertheless, they do not agree on $Z$ since $c\ne c\p $ and the set $Z\setminus (X\cup Y \cup W)$ is not empty. 
		That is,  the constraint $W\rightarrow Z$ is not satisfied which leads us to the conclusion that  $\{\Pi(f),\Pi(f\p)\}$ is inconsistent w.r.t $\Delta$.  
	\end{itemize}
\end{proof}

\end{document}